\documentclass[conference, a4paper]{IEEEtran}
\usepackage[utf8]{inputenc}
\usepackage{amsmath}
\usepackage{amsthm}
\usepackage{mathtools}
\usepackage{amsfonts}
\usepackage{hyperref}
\usepackage{url}
\usepackage{multiaudience}
  \SetNewAudience{conference}
  \SetNewAudience{arxiv}

\newcommand\blankfootnote[1]{%
  \let\svthefootnote\thefootnote%
  \let\thefootnote\relax\footnotetext{#1}%
  \let\thefootnote\svthefootnote%
}

\newcommand{\channelPMeasure}{Q}
\newcommand{\codebookPMeasure}{P}
\newcommand{\codebookRate}{R}
\newcommand{\channelIn}{X}
\newcommand{\channelInAlph}{\mathcal{X}}
\newcommand{\channelInAlphElement}{x}
\newcommand{\channelOut}{Y}
\newcommand{\channelOutAlph}{\mathcal{Y}}
\newcommand{\channelOutAlphElement}{y}
\newcommand{\channel}{\mathcal{W}}

\newcommand{\alphSubset}{A}
\newcommand{\codebook}{\mathcal{C}}
\newcommand{\codeword}{C}
\newcommand{\codewordIndex}{m}
\newcommand{\codebookBlocklength}{n}
\newcommand{\mutualInformation}[2]{I(#1;#2)}
\newcommand{\finalconstOne}{\gamma_1}
\newcommand{\finalconstTwo}{\gamma_2}
\newcommand{\totalvariation}[1]{\left\lVert #1 \right\rVert_\text{TV}}
\newcommand{\absolute}[1]{\left\lvert #1 \right\rvert}
\newcommand{\positive}[1]{\left[ #1 \right]^+}
\newcommand{\renyiParam}{\alpha}
\newcommand{\proofconstantOne}{{\beta_1}}
\newcommand{\proofconstantTwo}{{\beta_2}}

\newcommand{\generalAlphabetOne}{\mathcal{X}}
\newcommand{\generalAlphabetOneElement}{x}
\newcommand{\generalAlphabetTwo}{\mathcal{Y}}
\newcommand{\generalAlphabetTwoElement}{y}
\newcommand{\generalPMeasure}{\mu}
\newcommand{\generalPMeasureTwo}{\nu}
\newcommand{\generalRVOne}{X}
\newcommand{\generalRVTwo}{Y}
\newcommand{\RNDerivative}[2]{\frac{d{#1}}{d{#2}}}
\newcommand{\informationDensity}[2]{i({#1},{#2})}
\newcommand{\renyidiv}[3]{D_{#1}\left({#2} || {#3}\right)}
\newcommand{\Expectation}{\mathbb{E}}
\newcommand{\indicator}[1]{1_{#1}}
\newcommand{\lemmaconst}{\delta}
\newcommand{\lemmaexpectation}{\mu}
\newcommand{\chernoffParam}{\lambda}
\newcommand{\rateInfSubst}{r}
\newcommand{\distSubst}{a}
\newcommand{\intSubst}{b}
\newcommand{\derivwrt}[2]{\frac{d}{d {#2}} {#1}}
\newcommand{\generalReal}{t}
\newcommand{\typicalityParam}{\varepsilon}
\newcommand{\typicalSet}[1]{\mathcal{T}_{#1}}
\newcommand{\channelDispersion}{V}
\newcommand{\channelThirdMoment}{\rho}
\newcommand{\normalcdfComplement}{\mathcal{Q}}
\newcommand{\normalcdf}{\Phi}
\newcommand{\normalcdfComplementInverse}{\mathcal{Q}^{-1}}
\newcommand{\secondOrderParamC}{c}
\newcommand{\secondOrderParamD}{d}

\newcommand{\secondOrderTVdiff}{\xi}
\newcommand{\proofconstantlemma}{\xi}
\newcommand{\generalSummationIndex}{k}
\newcommand{\sigmaAlgebraIn}{{\mathcal{F}}}
\newcommand{\sigmaAlgebraOut}{{\mathcal{G}}}
\newcommand{\channelKernel}{K}
\newcommand{\channelDiscretizationIndex}{k}
\newcommand{\blocklengthIndex}{\ell}
\newcommand{\blockIndex}{j}
\newcommand{\blocklengthSubsequenceIndex}{i}
\newcommand{\powerset}{\mathcal{P}}

\newcommand{\sigmaAlgebraProduct}{\otimes}

\newcommand{\mutualInformationWrt}[3]{I_{#1}(#2;#3)}
\newcommand{\entropyWrt}[2]{H_{#1}(#2)}
\newcommand{\cardinality}[1]{\lvert #1 \rvert}

\newcommand{\reals}{\mathbb{R}}
\newcommand{\absolutelyContinuous}{\ll}
\newcommand{\generalSigmaAlgebra}{\mathcal{F}}
\newcommand{\generalSigmaAlgebraTwo}{\mathcal{G}}
\newcommand{\productSigmaAlgebra}{\otimes}
\newcommand{\kernelPower}[2]{{#1}^{\otimes {#2}}}
\newcommand{\sigmaAlgebraPower}[2]{{#1}^{\otimes {#2}}}
\newcommand{\generalPower}{n}

\newcommand{\capacityRegion}[2]{\mathcal{S}_{#1, #2}}
\newcommand{\inducingInputDistributions}[1]{G({#1})}
\newcommand{\generalFunction}{f}

\newtheorem{theorem}{Theorem}

\newtheorem{lemma}{Lemma}

\newtheorem{remark}{Remark}

\DefCurrentAudience{arxiv}
\showto{conference}{\allowdisplaybreaks}

\title{Resolvability on Continuous Alphabets}
\author{
Matthias Frey, Igor Bjelaković and Sławomir Stańczak
\\
Technische Universität Berlin
}

\begin{document}

\maketitle
\blankfootnote{
The work was supported by the German Research Foundation (DFG) under grant 
STA864/7-1 and by the German Federal Ministry of Education and Research under 
grant 16KIS0605.
}

\begin{abstract}
  We characterize the resolvability region for a large class of
  point-to-point channels with continuous alphabets. In our direct
  result, we prove not only the existence of good resolvability
  codebooks, but adapt an approach based on the Chernoff-Hoeffding
  bound to the continuous case showing that the probability of drawing
  an unsuitable codebook is doubly exponentially small. For the
  converse part, we show that our previous elementary result carries
  over to the continuous case easily under some mild continuity
  assumption.
\end{abstract}

\section{Introduction}

Channel resolvability has been established as an important tool in
information-theoretic security~\cite{CsiszarSecrecy,BlochStrongSecrecy,DevetakPrivateCapacity}. In particular, strong secrecy can be
derived directly from channel resolvability. The latter, roughly speaking,
is defined as the asymptotically smallest rate of a uniform random
seed that is needed to generate a channel input for which the channel
output well approximates a given target output under some suitable
approximation measure. Potential measures that are commonly used in
the literature are the Kullback-Leibler divergence and the variational
distance. In this paper, we focus on the latter one which is strong
enough to be related to security concepts from cryptography
\cite{FreyMACResolvabilityRegion}.

To the best of our knowledge, Wyner~\cite{WynerCommonInformation} was
the first to propose the problem of approximating a given output
distribution over a communication channel with as little randomness as
possible on the transmitter side. In~\cite{WynerCommonInformation}, he
used a normalized Kullback-Leibler divergence to measure the deviation
between the actual and the target output distribution. Han and
Verdú~\cite{HanApproximation} introduced the notion of channel
resolvability and formulated a similar problem except that they
assumed the variational distance as a metric. Unnormalized Kullback-Leibler divergence
was considered as a measure of similarity in~\cite{HayashiResolvability} and later in~\cite{HouInformational};
Rényi divergence was considered in~\cite{YuRenyi}. The results
of~\cite{CsiszarSecrecy,DevetakPrivateCapacity,CuffSoftCovering}
show that not only good resolvability codebooks exist but also that
the probability of drawing an unsuitable random codebook is doubly
exponentially small. Second order results for resolvability rate are
presented in~\cite{WatanabeSecondOrder,CuffSoftCovering} and
for MAC
in~\cite{FreyMACResolvability,FreyMACResolvabilityRegion}. Nonasymptotic
results are obtained in~\cite{HayashiResolvability}.

Converse theorems for arbitrary input and output distribution (without i.i.d.
assumption across channel uses) are contained in~\cite{HanApproximation,YagiResolvability}
and~\cite{SteinbergResolvability} for MAC. A converse resolvability result
based on Kullback-Leibler divergence is shown in~\cite{WynerCommonInformation}
and a simpler argument is given in~\cite{HouPhDThesis}. As we focus on variational
distance, these results do not carry over to our case. 

For MAC with finite alphabets and i.i.d. inputs, we have established
direct and converse results in~\cite{FreyMACResolvabilityRegion}. \emph{In
this work, we extend those results to continuous alphabets in the
point-to-point setting.} This extension to nondiscrete alphabets is in particular a
step towards dropping the assumption common in secrecy results that the wiretapper's alphabet
is discrete. Thus, these results may be extended to many channels particularly practically
relevant in wireless communications such as the AWGN or the Rayleigh fading channel.
We also remark that although we use the same technique as~\cite{CuffSoftCovering,FreyMACResolvability},
the extension to nondiscrete alphabets is not entirely straightforward: Dropping the assumption
that the alphabets are finite means that bounding the typical part of the variational distance with
applications of the Chernoff-Hoeffding and union bound is not possible. Instead, we apply the
Chernoff-Hoeffding bound in the usual way and infer the bound on the variation distance more directly
as layed out in Lemmas~\ref{lemma:typical} and \ref{lemma:typical-helper}.

\section{Preliminaries}
A \emph{channel} $\channel = ((\channelInAlph, \sigmaAlgebraIn), (\channelOutAlph, \sigmaAlgebraOut), \channelKernel)$ is given by an \emph{input alphabet} $\channelInAlph$ with $\sigma$-algebra $\sigmaAlgebraIn$, an \emph{output alphabet} $\channelOutAlph$ with $\sigma$-algebra $\sigmaAlgebraOut$ and a \emph{stochastic kernel} $\channelKernel$ which defines a stochastic transition between the input and output alphabets. I.e., $\channelKernel$ is a mapping from $\channelInAlph \times \sigmaAlgebraOut$ to $[0,1]$ such that $\channelKernel(\cdot, \alphSubset)$ is measurable for each $\alphSubset \in \sigmaAlgebraOut$ and $\channelKernel(\channelInAlphElement, \cdot)$ is a probability measure on $(\channelOutAlph, \sigmaAlgebraOut)$ for each $\channelInAlphElement \in \channelInAlph$. We assume throughout this paper that the input and output alphabets are Polish with Borel $\sigma$-algebra. $\channelIn$ and $\channelOut$ are random variables denoting the channel input and output respectively. Given $\sigma$-algebras $\generalSigmaAlgebra$ and $\generalSigmaAlgebraTwo$, we denote their product $\sigma$-algebra by $\generalSigmaAlgebra \productSigmaAlgebra \generalSigmaAlgebraTwo$. Likewise, the product of $\generalPower$ copies of $\generalSigmaAlgebra$ is denoted by $\sigmaAlgebraPower{\generalSigmaAlgebra}{\generalPower}$. The $\generalPower$th extension of $\channelKernel$ is given by
$
\kernelPower{\channelKernel}{\generalPower}(\channelInAlphElement^\generalPower, \bigtimes_{\blockIndex=1}^\generalPower \alphSubset_\blockIndex)
:=
\prod_{\blockIndex=1}^\generalPower \channelKernel(\channelInAlphElement_\blockIndex, \alphSubset_\blockIndex)
$.
An \emph{input distribution} $\channelPMeasure_\channelIn$ on $(\channelInAlph, \sigmaAlgebraIn)$ induces a joint distribution $\channelPMeasure_{\channelIn,\channelOut}$ on $(\channelInAlph \times \channelOutAlph, \sigmaAlgebraIn \productSigmaAlgebra \sigmaAlgebraOut)$ via
$
\channelPMeasure_{\channelIn,\channelOut}(\alphSubset_1 \times \alphSubset_2) := \int_{\alphSubset_1} \channelKernel(\channelInAlphElement, \alphSubset_2) \channelPMeasure_\channelIn(d\channelInAlphElement)
$
for $\alphSubset_1 \in \sigmaAlgebraIn$, $\alphSubset_2 \in \sigmaAlgebraOut$. The induced \emph{output distribution} is denoted $\channelPMeasure_\channelOut$. The $\codebookBlocklength$-fold products of the input and output distributions are denoted $\channelPMeasure_{\channelIn^\codebookBlocklength}$ and $\channelPMeasure_{\channelOut^\codebookBlocklength}$, respectively. A \emph{codebook} for the input alphabet $\channelInAlph$ with block length $\codebookBlocklength$ and rate $\codebookRate$ is a tuple $\codebook = (\codeword(\codewordIndex))_{\codewordIndex = 1}^{\exp(\codebookBlocklength\codebookRate)}$, where the \emph{codewords} are $\codeword(\codewordIndex) \in \channelInAlph^\codebookBlocklength$. We define the \emph{input distribution induced by $\codebook$} as
$
\codebookPMeasure_{\channelIn^\codebookBlocklength | \codebook}(\alphSubset) := \exp(-\codebookBlocklength\codebookRate) \sum_{\codewordIndex=1}^{\exp(\codebookBlocklength\codebookRate)} \indicator{\codeword(\codewordIndex) \in \alphSubset}.
$
$\codebookPMeasure_{\channelIn^\codebookBlocklength | \codebook}$ and $\kernelPower{\channelKernel}{\codebookBlocklength}$ induce an input-output distribution $\codebookPMeasure_{\channelIn^\codebookBlocklength, \channelOut^\codebookBlocklength | \codebook}$ and output distribution $\codebookPMeasure_{\channelOut^\codebookBlocklength | \codebook}$. Any $\channelPMeasure_\channelIn$ on $(\channelInAlph, \sigmaAlgebraIn)$ induces a distribution $\codebookPMeasure_\codebook$ on the set of possible codebooks by drawing all the components of all the codewords i.i.d. from $\channelInAlph$ according to $\channelPMeasure_\channelIn$. Given probability measures $\generalPMeasure$ and $\generalPMeasureTwo$ on $(\generalAlphabetOne, \generalSigmaAlgebra)$, we define the \emph{variational distance} as
$
\totalvariation{\generalPMeasure - \generalPMeasureTwo}
:=
\sup_{\alphSubset \in \generalSigmaAlgebra} (\generalPMeasure(\alphSubset) - \generalPMeasureTwo(\alphSubset))
$.
We say that $\generalPMeasure$ is absolutely continuous with respect to $\generalPMeasureTwo$, in symbols $\generalPMeasure \absolutelyContinuous \generalPMeasureTwo$, if all $\generalPMeasureTwo$-null sets are $\generalPMeasure$-null sets. If $\generalPMeasure \absolutelyContinuous \generalPMeasureTwo$, the Radon-Nikodym theorem states that there exists a measurable function $\RNDerivative{\generalPMeasure}{\generalPMeasureTwo}: \generalAlphabetOne \rightarrow [0,\infty)$, called the \emph{Radon-Nikodym derivative}, such that for every $\alphSubset \in \generalSigmaAlgebra$,
$
\generalPMeasure(\alphSubset) = \int\nolimits_{\alphSubset} \RNDerivative{\generalPMeasure}{\generalPMeasureTwo}(\generalAlphabetOneElement) \generalPMeasureTwo(d\generalAlphabetOneElement).
$
Given a channel and an input distribution, we define for any $\generalAlphabetOneElement \in \generalAlphabetOne$ and $\generalAlphabetTwoElement \in \generalAlphabetTwo$ the \emph{information density} of $(\generalAlphabetOneElement,\generalAlphabetTwoElement)$ as
$
\informationDensity{\generalAlphabetOneElement}{\generalAlphabetTwoElement}
:=
\log
  \RNDerivative{\channelKernel(\generalAlphabetOneElement,\cdot)}
               {\channelPMeasure_\channelOut}
    (\generalAlphabetTwoElement).
$
By convention, we say that the information density is $\infty$ on the singular set where $\channelKernel(\generalAlphabetOneElement,\cdot)$ is not absolutely continuous with respect to $\channelPMeasure_\channelOut$ and $-\infty$ where the relative density is $0$. Note that if the information density is finite almost everywhere, we can pick versions of the Radon-Nikodym derivatives such that $\informationDensity{\generalAlphabetOneElement}{\generalAlphabetTwoElement}$ is measurable with respect to $\sigmaAlgebraIn \sigmaAlgebraProduct \sigmaAlgebraOut$ \cite[Chap. 5, Theorem 4.44]{CinlarProbability}. In this case, we can define the \emph{mutual information} of $\generalRVOne$ and $\generalRVTwo$ as $\mutualInformation{\generalRVOne}{\generalRVTwo}:=\Expectation_{\channelPMeasure_{\channelIn, \channelOut}} \informationDensity{\generalAlphabetOneElement}{\generalAlphabetTwoElement}.$

\begin{shownto}{conference}
Due to lack of space, we omit some proofs and we also skip some details. For complete and more detailed proofs, we refer the reader to the extended version of this paper~\cite{arxivVersion}.
\end{shownto}

\section{Resolvability Region}
Given a channel $\channel = ((\channelInAlph, \sigmaAlgebraIn), (\channelOutAlph, \sigmaAlgebraOut), \channelKernel)$ and an output distribution $\channelPMeasure_\channelOut$, a rate $\codebookRate \in [0,\infty)$ is called \emph{achievable} if there is a sequence $(\codebook_\blocklengthIndex)_{\blocklengthIndex \geq 1}$ of codebooks with strictly increasing block lengths $\codebookBlocklength_\blocklengthIndex$ and rate $\codebookRate$ such that
$
\lim_{\blocklengthIndex \rightarrow \infty}
\totalvariation{
  \codebookPMeasure_{\channelOut^{\codebookBlocklength_\blocklengthIndex} | \codebook_{\blocklengthIndex}}
  -
  \channelPMeasure_{\channelOut^{\codebookBlocklength_\blocklengthIndex}}
}
=
0
$.
The \emph{resolvability region} $\capacityRegion{\channel}{\channelPMeasure_\channelOut}$ is the closure of the set of all achievable rates.
\begin{theorem}
Let $\channel = ((\channelInAlph, \sigmaAlgebraIn), (\channelOutAlph, \sigmaAlgebraOut), \channelKernel)$ be a channel such that $\channelInAlph$ is compact and for each $\alphSubset \subseteq \channelOutAlph$, $\channelInAlphElement \mapsto \channelKernel(\channelInAlphElement, \alphSubset)$ is a continuous mapping. Let $\channelPMeasure_\channelOut$ be an output distribution. Define
\begin{align*}
\inducingInputDistributions{\channelPMeasure_\channelOut}
:=
\{
  \channelPMeasure_\channelIn:
  &\channelPMeasure_\channelIn \text{ induces } \channelPMeasure_\channelOut \text { through } \channel,~
\\
  &\mutualInformationWrt{\channelPMeasure_{\channelIn, \channelOut}}{\channelIn}{\channelOut} < \infty
\}.
\end{align*}
Then
\begin{align*}
\capacityRegion{\channel}{\channelPMeasure_\channelOut}
=
\left\{
  \codebookRate \in \reals:
  \codebookRate \geq \inf_{\channelPMeasure_\channelIn \in \inducingInputDistributions{\channelPMeasure_\channelOut}} \mutualInformationWrt{\channelPMeasure_{\channelIn, \channelOut}}{\channelIn}{\channelOut}
\right\}.
\end{align*}
\end{theorem}
\begin{shownto}{conference}
The inclusion ``$\supseteq$'' is a direct consequence of Theorem~\ref{theorem:firstorder} and Remark~\ref{remark:firstorder} in Section~\ref{section:direct}. Theorem~\ref{theorem:firstorder} does not require the input alphabet to be compact and in many practically relevant cases, it even states that not only there exists a sequence of codebooks witnessing that the rate is achievable, but also that the probability of randomly drawing a ``bad'' codebook vanishes doubly exponentially  with increasing block length. The inclusion ``$\subseteq$'' is a consequence of Theorem~\ref{theorem:converse} proven in Section~\ref{section:converse}.
\end{shownto}
\begin{shownto}{arxiv}
The inclusion ``$\supseteq$'' is a direct consequence of Theorem~\ref{theorem:firstorder-weakened} in Section~\ref{section:direct}, which is a variation of Theorem~\ref{theorem:firstorder}, our main direct result. Theorems~\ref{theorem:firstorder} and~\ref{theorem:firstorder-weakened}  do not require the input alphabet to be compact and in many practically relevant cases, Theorem~\ref{theorem:firstorder} even states that not only there exists a sequence of codebooks witnessing that the rate is achievable, but also that the probability of randomly drawing a ``bad'' codebook vanishes doubly exponentially  with increasing block length. The inclusion ``$\subseteq$'' is a consequence of Theorem~\ref{theorem:converse} proven in Section~\ref{section:converse}.
\end{shownto}

\section{Direct Results}
\label{section:direct}
The main result of this section is the following theorem.
\begin{theorem}
\label{theorem:firstorder}
Given a channel $\channel = ((\channelInAlph, \sigmaAlgebraIn), (\channelOutAlph, \sigmaAlgebraOut), \channelKernel)$, an input distribution $\channelPMeasure_\channelIn$ such that the moment-generating function $\Expectation_{\channelPMeasure_{\channelIn, \channelOut}} \exp(\generalReal \cdot \informationDensity{\channelIn}{\channelOut})$ of the information density exists and is finite for some $\generalReal > 0$, and $\codebookRate > \mutualInformation{\channelIn}{\channelOut}$, there exist $\finalconstOne > 0$ and $\finalconstTwo > 0$ such that for large enough block lengths $\codebookBlocklength$, the randomized codebook distributions of block length $\codebookBlocklength$ and rate $\codebookRate$ satisfy
\begin{multline}
\label{thm:resdirect-probstatement}
\codebookPMeasure_\codebook \left(
  \totalvariation{
    \codebookPMeasure_{\channelOut^\codebookBlocklength | \codebook} - \channelPMeasure_{\channelOut^\codebookBlocklength}
  }
  >
  \exp(-\finalconstOne\codebookBlocklength)
\right)
\\
\leq
\exp\left(-\exp\left(\finalconstTwo\codebookBlocklength\right)\right).
\end{multline}
\end{theorem}
\showto{arxiv}{With a slight refinement of the proof, we can also establish the following second-order result.}
\showto{conference}{With a slight refinement of the proof, we can also establish the following second-order result. For details, we refer the reader to~\cite{arxivVersion}.}
\begin{theorem}
\label{theorem:secondorder}
Given a channel $\channel = ((\channelInAlph, \sigmaAlgebraIn), (\channelOutAlph, \sigmaAlgebraOut), \channelKernel)$, an input distribution $\channelPMeasure_\channelIn$ such that the information density $\informationDensity{\channelIn}{\channelOut}$ has finite central second moment $\channelDispersion$ and finite absolute third moment $\channelThirdMoment$, $\secondOrderTVdiff>0$ and $\secondOrderParamC>1$,
\showto{arxiv}{
suppose the rate $\codebookRate$ depends on $\codebookBlocklength$ in the following way:
\begin{align}
\label{def:secondorder-rate}
\codebookRate
=
\mutualInformation{\channelIn}{\channelOut}
+
\sqrt{\frac{\channelDispersion}{\codebookBlocklength}} \normalcdfComplementInverse(\secondOrderTVdiff)
+
\secondOrderParamC
\frac{\log \codebookBlocklength}{\codebookBlocklength},
\end{align}
}
\showto{conference}{
suppose that the rate is
$
\codebookRate
=
\mutualInformation{\channelIn}{\channelOut}
+
\sqrt{\channelDispersion/\codebookBlocklength} \normalcdfComplementInverse(\secondOrderTVdiff)
+
\secondOrderParamC
\log \codebookBlocklength/\codebookBlocklength,
$
}
where $\normalcdfComplement := 1 - \normalcdf$ with $\normalcdf$ the distribution function of the standard normal density. Then, for any $\secondOrderParamD \in (0, \secondOrderParamC-1)$ and $\codebookBlocklength$ that satisfy $\codebookBlocklength^{(\secondOrderParamC-\secondOrderParamD)/2} \geq 6$,
we have
\begin{multline}
\label{theorem:secondorder-probstatement}
\codebookPMeasure_\codebook \left(
  \totalvariation{
    \codebookPMeasure_{\channelOut^\codebookBlocklength | \codebook} - \channelPMeasure_{\channelOut^\codebookBlocklength}
  }
  >
  \lemmaexpectation\left(1+\frac{1}{\sqrt{\codebookBlocklength}}\right)
  +
  \frac{1}{\sqrt{\codebookBlocklength}}
\right)
\\
\begin{aligned}
\leq
&\exp\left(
  -\frac{1}{3}
  \codebookBlocklength
  \lemmaexpectation
  \exp(\codebookBlocklength\codebookRate)
\right)
\\
&+
\left(
  \frac{7}{6}
  +
  \sqrt{3\pi/2}
  \exp\left(
    \frac{3}{4}
  \right)
\right)
\exp\left(
  -\codebookBlocklength^{
    \frac{1}{2}
    (\secondOrderParamC - \secondOrderParamD - 1)
  }
\right),
\end{aligned}
\end{multline}
where
\showto{arxiv}{
\begin{align*}
\lemmaexpectation
:=
\normalcdfComplement\left(
  \normalcdfComplementInverse(\secondOrderTVdiff)
  +
  \secondOrderParamD
  \frac{\log \codebookBlocklength}{\sqrt{\codebookBlocklength\channelDispersion}}
\right)
+
\frac{\channelThirdMoment}
     {\channelDispersion^{\frac{3}{2}} \sqrt{\codebookBlocklength}}
\end{align*}
tends to $\secondOrderTVdiff$ for $\codebookBlocklength \rightarrow \infty$.
}
\showto{conference}{
$
\lemmaexpectation
:=
\normalcdfComplement\left(
  \normalcdfComplementInverse(\secondOrderTVdiff)
  +
  \secondOrderParamD
  \frac{\log \codebookBlocklength}{\sqrt{\codebookBlocklength\channelDispersion}}
\right)
+
\frac{\channelThirdMoment}
     {\channelDispersion^{\frac{3}{2}} \sqrt{\codebookBlocklength}}
$.
}
\end{theorem}

\showto{conference}{In order to prove Theorem~\ref{theorem:firstorder},}
\showto{arxiv}{In order to prove these theorems,}
given a codebook $\codebook$, we write the variational distance as
\begin{shownto}{arxiv}
\begin{align}
&\hphantom{{}={}}
\totalvariation{
  \codebookPMeasure_{\channelOut^\codebookBlocklength | \codebook} - \channelPMeasure_{\channelOut^\codebookBlocklength}
}
\nonumber
\\
&=
\sup\limits_{\alphSubset\in\sigmaAlgebraPower{\sigmaAlgebraOut}{\codebookBlocklength}} \left(
  \codebookPMeasure_{\channelOut^\codebookBlocklength | \codebook}(\alphSubset)
  -
  \channelPMeasure_{\channelOut^\codebookBlocklength}(\alphSubset)
\right)
\nonumber
\\
&=
\sup\limits_{\alphSubset\in\sigmaAlgebraPower{\sigmaAlgebraOut}{\codebookBlocklength}}
  \int\nolimits_\alphSubset \left(
    \RNDerivative{\codebookPMeasure_{\channelOut^\codebookBlocklength | \codebook}}{\channelPMeasure_{\channelOut^\codebookBlocklength}}(\channelOutAlphElement^\codebookBlocklength) - 1
  \right)
  \channelPMeasure_{\channelOut^\codebookBlocklength}(d\channelOutAlphElement^\codebookBlocklength)
\nonumber
\\
&=
\Expectation_{\channelPMeasure_{\channelOut^\codebookBlocklength}}
  \positive{ \RNDerivative{\codebookPMeasure_{\channelOut^\codebookBlocklength | \codebook}}{\channelPMeasure_{\channelOut^\codebookBlocklength}}(\channelOutAlphElement^\codebookBlocklength) - 1 }.
\label{totalvariation_positive}
\end{align}
\end{shownto}
\begin{shownto}{conference}
\begin{align}
&\hphantom{{}={}}
\totalvariation{
  \codebookPMeasure_{\channelOut^\codebookBlocklength | \codebook} - \channelPMeasure_{\channelOut^\codebookBlocklength}
}
\nonumber
=
\sup\limits_{\alphSubset\in\sigmaAlgebraPower{\sigmaAlgebraOut}{\codebookBlocklength}} \left(
  \codebookPMeasure_{\channelOut^\codebookBlocklength | \codebook}(\alphSubset)
  -
  \channelPMeasure_{\channelOut^\codebookBlocklength}(\alphSubset)
\right)
\nonumber
\\
&=
\sup\limits_{\alphSubset\in\sigmaAlgebraPower{\sigmaAlgebraOut}{\codebookBlocklength}}
  \int\nolimits_\alphSubset \left(
    \RNDerivative{\codebookPMeasure_{\channelOut^\codebookBlocklength | \codebook}}{\channelPMeasure_{\channelOut^\codebookBlocklength}}(\channelOutAlphElement^\codebookBlocklength) - 1
  \right)
  \channelPMeasure_{\channelOut^\codebookBlocklength}(d\channelOutAlphElement^\codebookBlocklength)
\nonumber
\\
&=
\Expectation_{\channelPMeasure_{\channelOut^\codebookBlocklength}}
  \positive{ \RNDerivative{\codebookPMeasure_{\channelOut^\codebookBlocklength | \codebook}}{\channelPMeasure_{\channelOut^\codebookBlocklength}}(\channelOutAlphElement^\codebookBlocklength) - 1 }.
\label{totalvariation_positive}
\end{align}
\end{shownto}


Note that throughout the proofs, we only consider codebooks $\codebook$ for which $\codebookPMeasure_{\channelOut^\codebookBlocklength | \codebook}$ is absolutely continuous with respect to $\channelPMeasure_{\channelOut^\codebookBlocklength}$. We can do this because the existence of a finite mutual information implies that $\channelKernel(\channelInAlphElement, \cdot)$ is absolutely continuous with respect to $\channelPMeasure_{\channelOut}$ for almost every $\channelInAlphElement$, and so the probability of drawing a codebook for which $\codebookPMeasure_{\channelOut^\codebookBlocklength | \codebook}$ is not absolutely continuous with respect to $\channelPMeasure_{\channelOut^\codebookBlocklength}$ is $0$. Similarly, we assume the existence of the other Radon-Nikodym derivatives that appear.

\showto{arxiv}{
We define the typical set
\begin{align}
\label{def:typicalset}
\typicalSet{\typicalityParam} := \left\{ (\channelInAlphElement^\codebookBlocklength,\channelOutAlphElement^\codebookBlocklength): \frac{1}{\codebookBlocklength} \informationDensity{\channelInAlphElement^\codebookBlocklength}{\channelOutAlphElement^\codebookBlocklength} \leq \mutualInformation{\channelIn}{\channelOut} + \typicalityParam \right\}
\end{align}
and split $\codebookPMeasure_{\channelOut^\codebookBlocklength | \codebook}$ into two measures
\begin{align}
\label{def:measure-typical-part}
&\begin{aligned}
&\codebookPMeasure_{1,\codebook}(\alphSubset)
:=
\exp(-\codebookBlocklength\codebookRate)
\\
&~~\cdot\sum\limits_{\codewordIndex=1}^{\exp(\codebookBlocklength\codebookRate)} \kernelPower{\channelKernel}{\codebookBlocklength}\left(
  \codeword(\codewordIndex),
  \alphSubset \cap \{ \channelOutAlphElement^\codebookBlocklength: (\codeword(\codewordIndex), \channelOutAlphElement^\codebookBlocklength) \in \typicalSet{\typicalityParam}\}
\right)
\end{aligned}
\\
\label{def:measure-atypical-part}
&\begin{aligned}
&\codebookPMeasure_{2,\codebook}(\alphSubset)
:=
\exp(-\codebookBlocklength\codebookRate)
\\
&~~\cdot\sum\limits_{\codewordIndex=1}^{\exp(\codebookBlocklength\codebookRate)} \kernelPower{\channelKernel}{\codebookBlocklength}\left(
  \codeword(\codewordIndex), \alphSubset \cap \{ \channelOutAlphElement^\codebookBlocklength: (\codeword(\codewordIndex), \channelOutAlphElement^\codebookBlocklength) \notin \typicalSet{\typicalityParam}\}
\right).
\end{aligned}
\end{align}
}
\showto{conference}{
We define the typical set
$
\typicalSet{\typicalityParam} := \{ (\channelInAlphElement^\codebookBlocklength,\channelOutAlphElement^\codebookBlocklength): \informationDensity{\channelInAlphElement^\codebookBlocklength}{\channelOutAlphElement^\codebookBlocklength}/\codebookBlocklength \leq \mutualInformation{\channelIn}{\channelOut} + \typicalityParam \}
$
and split $\codebookPMeasure_{\channelOut^\codebookBlocklength | \codebook}$ into two measures
$
\codebookPMeasure_{1,\codebook}(\alphSubset)
:=
\exp(-\codebookBlocklength\codebookRate)
\sum\nolimits_{\codewordIndex=1}^{\exp(\codebookBlocklength\codebookRate)} \kernelPower{\channelKernel}{\codebookBlocklength}(
  \codeword(\codewordIndex),
  \alphSubset \cap \{ \channelOutAlphElement^\codebookBlocklength: (\codeword(\codewordIndex), \channelOutAlphElement^\codebookBlocklength) \in \typicalSet{\typicalityParam}\}
)
$
and
$
\codebookPMeasure_{2,\codebook}(\alphSubset)
:=
\exp(-\codebookBlocklength\codebookRate)
\sum\nolimits_{\codewordIndex=1}^{\exp(\codebookBlocklength\codebookRate)} \kernelPower{\channelKernel}{\codebookBlocklength}(
  \codeword(\codewordIndex), \alphSubset \cap \{ \channelOutAlphElement^\codebookBlocklength: (\codeword(\codewordIndex), \channelOutAlphElement^\codebookBlocklength) \notin \typicalSet{\typicalityParam}\}
)
$.
}
We observe $\codebookPMeasure_{\channelOut^\codebookBlocklength | \codebook} = \codebookPMeasure_{1,\codebook} + \codebookPMeasure_{2,\codebook}$, which allows us to split~(\ref{totalvariation_positive}) into a typical and an atypical part
\begin{align}
&\hphantom{{}={}}
\totalvariation{\codebookPMeasure_{\channelOut^\codebookBlocklength | \codebook} - \channelPMeasure_{\channelOut^\codebookBlocklength}}
\nonumber
\\
&=
\Expectation_{\channelPMeasure_{\channelOut^\codebookBlocklength}}
  \positive{ \RNDerivative{\codebookPMeasure_{1,\codebook}}{\channelPMeasure_{\channelOut^\codebookBlocklength}}(\channelOutAlphElement^\codebookBlocklength) + \RNDerivative{\codebookPMeasure_{2,\codebook}}{\channelPMeasure_{\channelOut^\codebookBlocklength}}(\channelOutAlphElement^\codebookBlocklength) - 1 }
 \nonumber
\\
&\leq \label{split-typical-atypical}
\Expectation_{\channelPMeasure_{\channelOut^\codebookBlocklength}}
  \positive{
    \RNDerivative{\codebookPMeasure_{1,\codebook}}{\channelPMeasure_{\channelOut^\codebookBlocklength}}(\channelOutAlphElement^\codebookBlocklength)
    -
    1
  } 
+
\codebookPMeasure_{2,\codebook}(\channelOutAlph^\codebookBlocklength).
\end{align}
We next state and prove two lemmas that we will use as tools to bound the typical and atypical parts of this term separately.
\begin{lemma}[Bound for atypical terms]
\label{lemma:atypical}
Suppose
$\channelPMeasure_{\channelIn^\codebookBlocklength, \channelOut^\codebookBlocklength}
(\channelInAlph^\codebookBlocklength \times \channelOutAlph^\codebookBlocklength \setminus \typicalSet{\typicalityParam})
\leq
\lemmaexpectation$ and $\lemmaconst \in [0,1]$.
Then
\begin{align*}
\codebookPMeasure_\codebook\big(
  \codebookPMeasure_{2,\codebook}(\channelOutAlph^\codebookBlocklength) > \lemmaexpectation(1+\lemmaconst)
\big)
\leq
\exp\left(
  -\frac{1}{3} \lemmaconst^2 \lemmaexpectation \exp(\codebookBlocklength\codebookRate)
\right).
\end{align*}
\end{lemma}
\begin{shownto}{conference}
\begin{proof}
The lemma follows from the Chernoff-Hoeffding bound. A complete proof can be found in~\cite{arxivVersion}.
\end{proof}
\end{shownto}
\begin{shownto}{arxiv}
\begin{proof}
Observe
$\Expectation_{\codebookPMeasure_\codebook}(\codebookPMeasure_{2,\codebook}(\channelOutAlph^\codebookBlocklength))
=
\channelPMeasure_{\channelIn^\codebookBlocklength, \channelOut^\codebookBlocklength}
(\channelInAlph^\codebookBlocklength \times \channelOutAlph^\codebookBlocklength \setminus \typicalSet{\typicalityParam})
\leq
\lemmaexpectation$
and bound
\begin{align*}
&\hphantom{{}={}}
\codebookPMeasure_\codebook\left(\codebookPMeasure_{2,\codebook}(\channelOutAlph^\codebookBlocklength)
> 
\lemmaexpectation(1+\lemmaconst)
\right)
\\
&=
\codebookPMeasure_\codebook\left(
  \exp(\codebookBlocklength\codebookRate)\codebookPMeasure_{2,\codebook}(\channelOutAlph^\codebookBlocklength)
  >
  \lemmaexpectation
  \exp(\codebookBlocklength\codebookRate)
  (1+\lemmaconst)
\right) \\
&\begin{aligned}
=
\codebookPMeasure_\codebook\Bigg(
  &\sum\limits_{\codewordIndex=1}^{\exp(\codebookBlocklength\codebookRate)}
    \kernelPower{\channelKernel}{\codebookBlocklength}\left(
      \codeword(\codewordIndex),
      \{ \channelOutAlphElement^\codebookBlocklength: (\codeword(\codewordIndex), \channelOutAlphElement^\codebookBlocklength) \notin \typicalSet{\typicalityParam}\}
    \right)
\\
    &>
    \lemmaexpectation
    \exp(\codebookBlocklength\codebookRate)
    (1+\lemmaconst)
\Bigg)
\end{aligned}
\\
&\leq
\exp\left(
  -\frac{1}{3} \lemmaconst^2 \lemmaexpectation \exp(\codebookBlocklength\codebookRate)
\right).
\end{align*}
The inequality follows from the Chernoff-Hoeffding bound~\cite[Ex. 1.1]{ConcentrationTextbook} by noting that we sum probabilities (i.e. values in $[0,1]$) on the left side, that these probabilities are independently distributed under $\codebookPMeasure_\codebook$ and that by the hypothesis of the lemma the expectation of the term on the left is bounded by $\lemmaexpectation\exp(\codebookBlocklength\codebookRate)$.
\end{proof}
\end{shownto}
\begin{lemma}[Bound for typical terms]
\label{lemma:typical}
Let $\lemmaconst,\chernoffParam > 0$ and define
\begin{align}
\label{proof:resdirect-subst-rateInf}
\rateInfSubst &:= \exp(\codebookBlocklength(\codebookRate-\mutualInformation{\channelIn}{\channelOut}-\typicalityParam)).
\end{align}
Suppose $\rateInfSubst/(6\chernoffParam) \geq 1$. Then
\begin{multline}
\label{lemma:typicalbound-probstatement}
\codebookPMeasure_\codebook\left(
  \Expectation_{\channelPMeasure_{\channelOut^\codebookBlocklength}}
    \positive{
      \RNDerivative{\codebookPMeasure_{1,\codebook}}{\channelPMeasure_{\channelOut^\codebookBlocklength}}(\channelOutAlphElement^\codebookBlocklength)
      -
      1
    } 
  >
  \lemmaconst
\right)
\\
\leq
\Bigg(
  1
  +
  \sqrt{\frac{3\pi}{2}}
  \exp\left(
    \frac{3\chernoffParam^2}{4\rateInfSubst}
  \right)
  \frac{\chernoffParam}{\sqrt{\rateInfSubst}}
  +
  \exp(-\chernoffParam)
\Bigg)
\exp(-\lemmaconst\chernoffParam).
\end{multline}
\end{lemma}
Before we prove this lemma, we make an observation that we need in the proof.
\begin{lemma}
\label{lemma:typical-helper}
Let $\generalFunction$ be a measurable function mapping codebooks and elements of $\channelOutAlph^\codebookBlocklength$ to the nonnegative reals and let $\chernoffParam, \lemmaconst > 0$. Then
\begin{multline}
\label{lemma:typical-helper-probstatement}
\hat{\codebookPMeasure}
:=
\codebookPMeasure_\codebook\left(
  \Expectation_{\channelPMeasure^\codebookBlocklength}
    \generalFunction(\codebook, \channelOutAlphElement^\codebookBlocklength)
  >
  \lemmaconst
\right)
\\
\leq
\Expectation_{\channelPMeasure^\codebookBlocklength}
  \Expectation_{\codebookPMeasure_\codebook} \Big(
    \exp(\chernoffParam\generalFunction(\codebook, \channelOutAlphElement^\codebookBlocklength))
  \Big)
\exp(-\lemmaconst\chernoffParam).
\end{multline}
\end{lemma}
\begin{shownto}{arxiv}
\begin{proof}
An application of the Chernoff bound yields
\[
\hat{\codebookPMeasure}
\leq
\Expectation_{\codebookPMeasure_\codebook} \left(
\exp\left(
  \chernoffParam
  \Expectation_{\channelPMeasure^\codebookBlocklength}
    \generalFunction(\codebook, \channelOutAlphElement^\codebookBlocklength)
\right)
\right)
\exp(-\lemmaconst\chernoffParam).
\]
We can then prove (\ref{lemma:typical-helper-probstatement}) by successive applications of Jensen's inequality and Fubini's theorem.
\end{proof}
\end{shownto}
\begin{shownto}{conference}
\begin{proof}
An application of the Chernoff bound yields
$
\hat{\codebookPMeasure}
\leq
\Expectation_{\codebookPMeasure_\codebook} \left(
\exp\left(
  \chernoffParam
  \Expectation_{\channelPMeasure^\codebookBlocklength}
    \generalFunction(\codebook, \channelOutAlphElement^\codebookBlocklength)
\right)
\right)
\exp(-\lemmaconst\chernoffParam)
$.
We then prove (\ref{lemma:typical-helper-probstatement}) by applying Jensen's inequality and Fubini's theorem.
\end{proof}
\end{shownto}
\begin{proof}[Proof of Lemma~\ref{lemma:typical}]
We begin by examining parts of the term in (\ref{lemma:typicalbound-probstatement}) for fixed, but arbitrary $\codebook$ and $\channelOutAlphElement^\codebookBlocklength$ and rewrite
\begin{multline*}
\rateInfSubst \RNDerivative{\codebookPMeasure_{1,\codebook}}{\channelPMeasure_{\channelOut^\codebookBlocklength}}(\channelOutAlphElement^\codebookBlocklength)
=
\sum\limits_{\codewordIndex=1}^{\exp(\codebookBlocklength\codebookRate)}
  \exp\left(\codebookBlocklength(-\mutualInformation{\channelIn}{\channelOut} - \typicalityParam)\right) 
\\
\cdot
\RNDerivative{\kernelPower{\channelKernel}{\codebookBlocklength}(\codeword(\codewordIndex), \cdot)}{\channelPMeasure_{\channelOut^\codebookBlocklength}}(\channelOutAlphElement^\codebookBlocklength)
\indicator{(\codeword(\codewordIndex), \channelOutAlphElement^\codebookBlocklength) \in \typicalSet{\typicalityParam}}.
\end{multline*}
\begin{shownto}{arxiv}
Now, we observe that the indicator function bounds the relative density to be at most $\exp(\codebookBlocklength(\mutualInformation{\channelIn}{\channelOut}+\typicalityParam))$ and thus every term in the sum to range within $[0,1]$ and that furthermore
\begin{align*}
\Expectation_{\codebookPMeasure_\codebook}\left( \rateInfSubst \RNDerivative{\codebookPMeasure_{1,\codebook}}{\channelPMeasure_{\channelOut^\codebookBlocklength}}(\channelOutAlphElement^\codebookBlocklength)\right)
&\leq
\exp\left(\codebookBlocklength(-\mutualInformation{\channelIn}{\channelOut} - \typicalityParam)\right)
\\
&\cdot
\sum\limits_{\codewordIndex=1}^{\exp(\codebookBlocklength\codebookRate)}
  \Expectation_{\codebookPMeasure_\codebook}\left( \RNDerivative{\kernelPower{\channelKernel}{\codebookBlocklength}(\codeword(\codewordIndex), \cdot)}{\channelPMeasure_{\channelOut^\codebookBlocklength}}(\channelOutAlphElement^\codebookBlocklength)\right) 
\\
&=
\rateInfSubst.
\end{align*}
\end{shownto}
\begin{shownto}{conference}
Now, we observe that the indicator function bounds every term in the sum to range within $[0,1]$ and that furthermore
\begin{multline*}
\Expectation_{\codebookPMeasure_\codebook}\left( \rateInfSubst \RNDerivative{\codebookPMeasure_{1,\codebook}}{\channelPMeasure_{\channelOut^\codebookBlocklength}}(\channelOutAlphElement^\codebookBlocklength)\right)
\leq
\exp\left(\codebookBlocklength(-\mutualInformation{\channelIn}{\channelOut} - \typicalityParam)\right)
\\
\cdot
\sum\limits_{\codewordIndex=1}^{\exp(\codebookBlocklength\codebookRate)}
  \Expectation_{\codebookPMeasure_\codebook}\left( \RNDerivative{\kernelPower{\channelKernel}{\codebookBlocklength}(\codeword(\codewordIndex), \cdot)}{\channelPMeasure_{\channelOut^\codebookBlocklength}}(\channelOutAlphElement^\codebookBlocklength)\right) 
=
\rateInfSubst.
\end{multline*}
\end{shownto}
We then use these observations to yield, for any $\proofconstantlemma > 0$,
\begin{shownto}{arxiv}
\begin{align}
&\phantom{{}={}}
\codebookPMeasure_\codebook \left(
  \exp\left(
    \chernoffParam
    \positive{
      \RNDerivative{\codebookPMeasure_{1,\codebook}}{\channelPMeasure_{\channelOut^\codebookBlocklength}}(\channelOutAlphElement^\codebookBlocklength)
      -
      1
    }
  \right)
  >
  \exp(\chernoffParam\proofconstantlemma)
\right)
\nonumber
\\
\label{proof:resdirect-equalevents}
&=
\codebookPMeasure_\codebook \left(
  \RNDerivative{\codebookPMeasure_{1,\codebook}}{\channelPMeasure_{\channelOut^\codebookBlocklength}}(\channelOutAlphElement^\codebookBlocklength)
  >
  1 + \proofconstantlemma
\right)
\\
&\begin{aligned}
=
\codebookPMeasure_\codebook \bigg(
  \rateInfSubst
  \RNDerivative{\codebookPMeasure_{1,\codebook}}{\channelPMeasure_{\channelOut^\codebookBlocklength}}(\channelOutAlphElement^\codebookBlocklength)
  >
  \left(
    1 + \proofconstantlemma
  \right)
  \rateInfSubst
\bigg)
\end{aligned}
\nonumber
\\
&\leq \label{proof:resdirect-typical-chhoeff}
\exp \left(
  -\frac{\proofconstantlemma^2}{2\left(1+\frac{\proofconstantlemma}{3}\right)} \rateInfSubst
\right),
\end{align}
\end{shownto}
\begin{shownto}{conference}
\begin{align}
&\phantom{{}={}}
\codebookPMeasure_\codebook \left(
  \exp\left(
    \chernoffParam
    \positive{
      \RNDerivative{\codebookPMeasure_{1,\codebook}}{\channelPMeasure_{\channelOut^\codebookBlocklength}}(\channelOutAlphElement^\codebookBlocklength)
      -
      1
    }
  \right)
  >
  \exp(\chernoffParam\proofconstantlemma)
\right)
\nonumber
\\
\label{proof:resdirect-equalevents}
&=
\codebookPMeasure_\codebook \left(
  \RNDerivative{\codebookPMeasure_{1,\codebook}}{\channelPMeasure_{\channelOut^\codebookBlocklength}}(\channelOutAlphElement^\codebookBlocklength)
  >
  1 + \proofconstantlemma
\right)
\\
&\begin{aligned}
=
\codebookPMeasure_\codebook \bigg(
  \rateInfSubst
  \RNDerivative{\codebookPMeasure_{1,\codebook}}{\channelPMeasure_{\channelOut^\codebookBlocklength}}(\channelOutAlphElement^\codebookBlocklength)
  >
  \left(
    1 + \proofconstantlemma
  \right)
  \rateInfSubst
\bigg)
\end{aligned}
\leq
\label{proof:resdirect-typical-chhoeff}
\exp \left(
  -\frac{\proofconstantlemma^2}{2\left(1+\frac{\proofconstantlemma}{3}\right)} \rateInfSubst
\right),
\end{align}
\end{shownto}

where (\ref{proof:resdirect-equalevents}) holds because the two measured events are equal and (\ref{proof:resdirect-typical-chhoeff}) follows by the Chernoff-Hoeffding bound~\cite[Theorem 2.3b]{McDiarmidConcentration}. (\ref{proof:resdirect-typical-chhoeff}) can be upper bounded by
\begin{shownto}{arxiv}
\begin{align}
\label{proof:resdirect-typical-lessthanone}
\exp \left(
  -\frac{\proofconstantlemma^2}{3} \rateInfSubst
\right)
\end{align}
for $\proofconstantlemma \leq 1$ (in particular) and by
\begin{align}
\label{proof:resdirect-typical-greaterone}
\exp \left(
  -\frac{\proofconstantlemma}{3} \rateInfSubst
\right)
\end{align}
for $\proofconstantlemma \geq 1$ (in particular).
\end{shownto}
\begin{shownto}{conference}
\begin{align}
\label{proof:resdirect-typical}
\exp \left(
  -\frac{\proofconstantlemma^2}{3} \rateInfSubst
\right)
\text{ for }
\proofconstantlemma \leq 1;~
\exp \left(
  -\frac{\proofconstantlemma}{3} \rateInfSubst
\right)
\text{ for }
\proofconstantlemma \geq 1.
\end{align}
\end{shownto}
We will in the following use the substitutions
\showto{arxiv}{
\begin{align}
\label{proof:resdirect-subst-dist}
\distSubst &:= \exp(\chernoffParam\proofconstantlemma)
\\
\label{proof:resdirect-subst-int}
\intSubst &:= \frac{\log(\distSubst)}{\chernoffParam} \sqrt{\frac{2\rateInfSubst}{3}} - \sqrt{\frac{3}{2\rateInfSubst}} \chernoffParam.
\end{align}
Since we will be using (\ref{proof:resdirect-subst-int}) for integration by substitution, we note that it implies
\begin{align}
\derivwrt{\distSubst}{\intSubst}
=
\exp\left(
  \intSubst \chernoffParam \sqrt{\frac{3}{2\rateInfSubst}}+\chernoffParam^2\frac{3}{2\rateInfSubst}
\right)
\chernoffParam \sqrt{\frac{3}{2\rateInfSubst}}.
\end{align}
}
\showto{conference}{
\begin{align}
\label{proof:resdirect-subst}
\distSubst := \exp(\chernoffParam\proofconstantlemma)
,~~
\intSubst := \frac{\log(\distSubst)}{\chernoffParam} \sqrt{\frac{2\rateInfSubst}{3}} - \sqrt{\frac{3}{2\rateInfSubst}} \chernoffParam.
\end{align}
}
We have, e.g. by~\cite[Eq. 21.9]{BillingsleyProbability},
\begin{align}
&\phantom{{}={}}
\Expectation_{\codebookPMeasure_\codebook} \left(
  \exp\left(
    \chernoffParam
    \positive{
      \RNDerivative{\codebookPMeasure_{1,\codebook}}{\channelPMeasure_{\channelOut^\codebookBlocklength}}(\channelOutAlphElement^\codebookBlocklength)
      -
      1
    }
  \right)
\right)
\nonumber
\\
\label{proof:resdirect-exprepres}
&=
\int\nolimits_0^\infty
  \codebookPMeasure_\codebook \left(
    \exp\left(
      \chernoffParam
      \positive{
        \RNDerivative{\codebookPMeasure_{1,\codebook}}{\channelPMeasure_{\channelOut^\codebookBlocklength}}(\channelOutAlphElement^\codebookBlocklength)
        -
        1
      }
    \right)
    >
    \distSubst
  \right)
d\distSubst
\end{align}
and upper bound this integral by splitting the integration domain into three parts: The integration over $[0,1]$ can be upper bounded by $1$ (since the integrand is a probability). The integration over $[1,\exp(\chernoffParam)]$ can be upper bounded by
\begin{shownto}{arxiv}
\begin{align}
\label{proof:resdirect-substitution}
&\hphantom{{}={}}
\begin{aligned}
\int\nolimits_1^{\infty}
  \exp \left(
  -\frac{(\log \distSubst)^2}{3\chernoffParam^2} \rateInfSubst
\right)
d\distSubst
\end{aligned}
\\
\label{proof:resdirect-intsubst}
&\begin{aligned}
=
\int\nolimits_{0}^\infty
  \exp \Bigg(
  &-
  \frac{\intSubst^2 \chernoffParam^2 \frac{3}{2\rateInfSubst}
         +
         2 \intSubst \chernoffParam^3
         \left(
           \frac{3}{2\rateInfSubst}
         \right)^{\frac{3}{2}}
         +
         \chernoffParam^4 \left(\frac{3}{2\rateInfSubst}\right)^2
        }
        {3\chernoffParam^2}
  \rateInfSubst
  \\
  &+
  \intSubst \chernoffParam \sqrt{\frac{3}{2\rateInfSubst}}
  +
  \chernoffParam^2\frac{3}{2\rateInfSubst}
\Bigg)
\chernoffParam \sqrt{\frac{3}{2\rateInfSubst}}
d\intSubst
\end{aligned}
\\
\label{proof:resdirect-intfinal}
&=
\int\nolimits_{0}^\infty
  \exp \left(
    - \frac{\intSubst^2}{2}
  \right)
d\intSubst
\cdot
\exp\left(
  \frac{3\chernoffParam^2}{4\rateInfSubst}
\right)
\chernoffParam \sqrt{\frac{3}{2\rateInfSubst}}.
\end{align}
(\ref{proof:resdirect-substitution}) follows by substituting (\ref{proof:resdirect-typical-lessthanone}) as well as (\ref{proof:resdirect-subst-dist}) and enlarging the integration domain to $[1,\infty)$, which can be done because the integrand is nonnegative. (\ref{proof:resdirect-intsubst}) follows by the rule for integration by substitution using (\ref{proof:resdirect-subst-int}).
\end{shownto}
\begin{shownto}{conference}
\begin{align}
\label{proof:resdirect-substitution}
&\hphantom{{}={}}
\begin{aligned}
\int\nolimits_1^{\infty}
  \exp \left(
  -\frac{(\log \distSubst)^2}{3\chernoffParam^2} \rateInfSubst
\right)
d\distSubst
\end{aligned}
\\
\label{proof:resdirect-intfinal}
&=
\int\nolimits_{0}^\infty
  \exp \left(
    - \frac{\intSubst^2}{2}
  \right)
d\intSubst
\cdot
\exp\left(
  \frac{3\chernoffParam^2}{4\rateInfSubst}
\right)
\chernoffParam \sqrt{\frac{3}{2\rateInfSubst}}.
\end{align}
(\ref{proof:resdirect-substitution}) follows by substituting (\ref{proof:resdirect-typical}) as well as (\ref{proof:resdirect-subst}) and enlarging the integration domain to $[1,\infty)$. (\ref{proof:resdirect-intfinal}) follows by the rule for integration by substitution using (\ref{proof:resdirect-subst}).
\end{shownto}

The integration over $[\exp(\chernoffParam), \infty)$ can be upper bounded by
\begin{align}
\label{proof:resdirect-substitution-large}
&\hphantom{{}={}}
\int\nolimits_{\exp(\chernoffParam)}^{\infty}
  \exp \left(
  -\frac{\log \distSubst}{3\chernoffParam} \rateInfSubst
\right)
d\distSubst
=
\int\nolimits_{\exp(\chernoffParam)}^{\infty}
\distSubst^{
  -\rateInfSubst/(3\chernoffParam)
}
d\distSubst
\\
\label{proof:resdirect-intfinal-large}
&=
\frac{\exp(\chernoffParam (1-\rateInfSubst/(3\chernoffParam)))}
     {\rateInfSubst/(3\chernoffParam)-1}
\leq
6\frac{\chernoffParam}{\rateInfSubst} \exp\left(-\frac{\rateInfSubst}{6}\right)
\leq
\exp(-\chernoffParam),
\end{align}
where (\ref{proof:resdirect-substitution-large}) is by (\showto{arxiv}{\ref{proof:resdirect-typical-greaterone}}\showto{conference}{\ref{proof:resdirect-typical}}) and the inequalities are true because $\rateInfSubst/(6\chernoffParam) \geq 1$.
We now apply Lemma~\ref{lemma:typical-helper} with
$
\generalFunction(\codebook, \channelOutAlphElement^\codebookBlocklength)
:=
\positive{
      \RNDerivative{\codebookPMeasure_{1,\codebook}}{\channelPMeasure_{\channelOut^\codebookBlocklength}}(\channelOutAlphElement^\codebookBlocklength)
      -
      1
}
$.
In the resulting bound, we substitute the bound of $1$ for integration domain $[0,1]$ as well as (\ref{proof:resdirect-intfinal}) and (\ref{proof:resdirect-intfinal-large}), substitute back (\ref{proof:resdirect-subst-rateInf}) and note that $\exp(-\intSubst^2/2)$ is the well-known unnormalized standard normal density, and get~(\ref{lemma:typicalbound-probstatement}).
\end{proof}

\begin{proof}[Proof of Theorem~\ref{theorem:firstorder}]
\begin{shownto}{arxiv}
In order to bound the atypical term in the sum~(\ref{split-typical-atypical}), note first that for any $\renyiParam>1$,
\begin{align}
&\hphantom{{}={}}
\channelPMeasure_{\channelIn^\codebookBlocklength, \channelOut^\codebookBlocklength}(
  \channelInAlph^\codebookBlocklength \times \channelOutAlph^\codebookBlocklength \setminus \typicalSet{\typicalityParam}
) \nonumber
\\
&=
\channelPMeasure_{\channelIn^\codebookBlocklength, \channelOut^\codebookBlocklength}\left(
  \left\{
    (\channelInAlphElement^\codebookBlocklength,\channelOutAlphElement^\codebookBlocklength)
    :
    \informationDensity{\channelInAlphElement^\codebookBlocklength}{\channelOutAlphElement^\codebookBlocklength}/\codebookBlocklength
    >
    \mutualInformation{\channelIn}{\channelOut} + \typicalityParam
  \right\}
\right) \nonumber
\\
&\begin{aligned}
=
\channelPMeasure_{\channelIn^\codebookBlocklength, \channelOut^\codebookBlocklength}\big(
  &\{
    (\channelInAlphElement^\codebookBlocklength,\channelOutAlphElement^\codebookBlocklength)
    :
    \exp\left((\renyiParam-1) \informationDensity{\channelInAlphElement^\codebookBlocklength}{\channelOutAlphElement^\codebookBlocklength}\right)
    \\
    &>
    \exp\left(
      (\renyiParam-1)\codebookBlocklength\left(\mutualInformation{\channelIn}{\channelOut} + \typicalityParam \right)
    \right)
  \}
\big)
\end{aligned}
\nonumber \\
&\begin{aligned}
\leq
\label{atypical-bound-markov}
&\int\nolimits_{\channelInAlph^\codebookBlocklength \times \channelOutAlph^\codebookBlocklength}
  \exp\left(
    (\renyiParam-1)\informationDensity{\channelInAlphElement^\codebookBlocklength}{\channelOutAlphElement^\codebookBlocklength}
  \right)
  \channelPMeasure_{\channelIn, \channelOut}(d(\channelInAlphElement^\codebookBlocklength, \channelOutAlphElement^\codebookBlocklength))
  \\
  &\cdot
  \exp\left(-(\renyiParam-1)\codebookBlocklength\left(\mutualInformation{\channelIn}{\channelOut} + \typicalityParam \right)
\right)
\end{aligned}
\\
&\begin{aligned}
&\begin{aligned}
=
\exp \log \Bigg(
  \int\nolimits_{\channelInAlph^\codebookBlocklength \times \channelOutAlph^\codebookBlocklength}
    &\left(
      \RNDerivative{\kernelPower{\channelKernel}{\codebookBlocklength}(\codeword(\codewordIndex), \cdot)}{\channelPMeasure_{\channelOut^\codebookBlocklength}}
      (\channelOutAlphElement^\codebookBlocklength)
    \right)^{\renyiParam-1}
    \\
    &\cdot \channelPMeasure_{\channelIn, \channelOut}(d(\channelInAlphElement^\codebookBlocklength, \channelOutAlphElement^\codebookBlocklength))
\Bigg)
\end{aligned}
\\
&~~~\cdot
\exp\left(
  -\codebookBlocklength(\renyiParam-1)\left(\mutualInformation{\channelIn}{\channelOut} + \typicalityParam \right)
\right) \nonumber
\end{aligned}
\\
&=
\label{atypical-bound-mutualInformation}
\exp \left(
  -\codebookBlocklength (\renyiParam-1)
  \left(
    \mutualInformation{\channelIn}{\channelOut} + \typicalityParam - \renyidiv{\renyiParam}{\channelPMeasure_{\channelIn, \channelOut}}{\channelPMeasure_\channelIn \channelPMeasure_\channelOut}
  \right)
\right)
\\
&\leq
\label{atypical-bound-final-result}
\exp(-\codebookBlocklength\proofconstantOne),
\end{align}
where~(\ref{atypical-bound-markov}) follows by applying Markov's inequality and~(\ref{atypical-bound-final-result}) as long as
\begin{align}
\label{proof:resdirect-renyidiv-bound}
\proofconstantOne \leq (\renyiParam-1)\left(\mutualInformation{\channelIn}{\channelOut} + \typicalityParam - \renyidiv{\renyiParam}{\channelPMeasure_{\channelIn, \channelOut}}{\channelPMeasure_\channelIn \channelPMeasure_\channelOut} \right).
\end{align}
\end{shownto}
\begin{shownto}{conference}
In order to bound the atypical term in the sum~(\ref{split-typical-atypical}), note first that an application of the Chernoff bound and some elementary manipulations yield, for any $\renyiParam>1$,
\begin{align}
&\hphantom{{}={}}
\channelPMeasure_{\channelIn^\codebookBlocklength, \channelOut^\codebookBlocklength}(
  \channelInAlph^\codebookBlocklength \times \channelOutAlph^\codebookBlocklength \setminus \typicalSet{\typicalityParam}
) \nonumber
\\
&\leq
\exp \left(
  -\codebookBlocklength (\renyiParam-1)
  \left(
    \mutualInformation{\channelIn}{\channelOut} + \typicalityParam - \renyidiv{\renyiParam}{\channelPMeasure_{\channelIn, \channelOut}}{\channelPMeasure_\channelIn \channelPMeasure_\channelOut}
  \right)
\right)
\nonumber
\\
&\leq
\label{atypical-bound-final-result}
\exp(-\codebookBlocklength\proofconstantOne),
\end{align}
provided
$
\proofconstantOne \leq (\renyiParam-1)\left(\mutualInformation{\channelIn}{\channelOut} + \typicalityParam - \renyidiv{\renyiParam}{\channelPMeasure_{\channelIn, \channelOut}}{\channelPMeasure_\channelIn \channelPMeasure_\channelOut} \right)
$.\\
\end{shownto}

Note that since the moment-generating function $\Expectation_{\channelPMeasure_{\channelIn, \channelOut}} \exp(\generalReal \cdot \informationDensity{\channelIn}{\channelOut})$ exists and is finite for some $\generalReal > 0$, there is some $\renyiParam' > 1$ such that $\renyidiv{\renyiParam'}{\channelPMeasure_{\channelIn, \channelOut}}{\channelPMeasure_\channelIn \channelPMeasure_\channelOut}$ is finite, and thus $\renyidiv{\renyiParam}{\channelPMeasure_{\channelIn, \channelOut}}{\channelPMeasure_\channelIn \channelPMeasure_\channelOut}$ is finite and continuous in $\renyiParam$ for $\renyiParam \leq \renyiParam'$~\cite{RenyiDiv}. Since $\renyidiv{\renyiParam}{\channelPMeasure_{\channelIn, \channelOut}}{\channelPMeasure_\channelIn \channelPMeasure_\channelOut} \rightarrow \mutualInformation{\channelIn}{\channelOut}$ for $\renyiParam \rightarrow 1$, we can choose $\renyiParam > 1$, but sufficiently close to $1$ such that the bound on $\proofconstantOne$ is positive.

We can now apply Lemma~\ref{lemma:atypical} with $\lemmaexpectation := \exp(-\codebookBlocklength\proofconstantOne)$ and $\lemmaconst := 1$ and get
\begin{multline}
\codebookPMeasure_\codebook\big(
  \codebookPMeasure_{2,\codebook}(\channelOutAlph^\codebookBlocklength)
  >
  2\exp(-\codebookBlocklength\proofconstantOne)
\big)
\\
\leq
\exp\left(
  -\frac{1}{3}
  \exp(\codebookBlocklength(\codebookRate-\proofconstantOne))
\right).
\label{proof:firstorder-atypical-probability-statement}
\end{multline}

To bound the typical term in~(\ref{split-typical-atypical}), we apply Lemma~\ref{lemma:typical} with $\chernoffParam := \exp(\codebookBlocklength\proofconstantTwo)$ and $\lemmaconst := \exp(-\codebookBlocklength\proofconstantOne)$, which yields
\begin{multline}
\codebookPMeasure_\codebook \left(
  \Expectation_{\channelPMeasure_{\channelOut^\codebookBlocklength}}
    \positive{
      \RNDerivative{\codebookPMeasure_{1,\codebook}}{\channelPMeasure_{\channelOut^\codebookBlocklength}}(\channelOutAlphElement^\codebookBlocklength)
      -
      1
    } 
>
\exp(-\codebookBlocklength\proofconstantOne)
\right)
\leq
\\
\label{typical-bound-final-result}
\Bigg(
  1
  +
  \sqrt{\frac{3 \pi}{2}}
  \exp\left(
    \frac{3}{4}
    \exp(-\codebookBlocklength(\codebookRate-\mutualInformation{\channelIn}{\channelOut} - \typicalityParam - 2\proofconstantTwo) )
  \right)
  \\
  \hphantom{{}+{}}\cdot\exp\left(
    -\frac{1}{2}\codebookBlocklength(\codebookRate-\mutualInformation{\channelIn}{\channelOut} - \typicalityParam - 2\proofconstantTwo)
  \right)
  \\
  +
  \exp(-\exp(\codebookBlocklength\proofconstantTwo))
\Bigg)
\cdot\exp\left(
  -
  \exp(\codebookBlocklength(\proofconstantTwo-\proofconstantOne))
\right),
\end{multline}
as long as $\codebookBlocklength$ is sufficiently large such that $\exp(\codebookBlocklength(\codebookRate-\mutualInformation{\channelIn}{\channelOut} - \typicalityParam))/6 \geq 1$.

We are now ready to put everything together:
\begin{shownto}{arxiv}
Considering (\ref{split-typical-atypical}), (\ref{proof:firstorder-atypical-probability-statement}) and (\ref{typical-bound-final-result}), an application of the union bound yields the sum of (\ref{proof:firstorder-atypical-probability-statement}) and (\ref{typical-bound-final-result}) as an upper bound for
$
\codebookPMeasure_\codebook \left(
  \totalvariation{\codebookPMeasure_{\channelOut^\codebookBlocklength | \codebook} - \channelPMeasure_{\channelOut^\codebookBlocklength}}
  >
  3\exp(-\codebookBlocklength\proofconstantOne)
\right)
$.

We choose $\typicalityParam < \codebookRate - \mutualInformation{\channelIn}{\channelOut}$, then $\proofconstantOne < (\codebookRate - \mutualInformation{\channelIn}{\channelOut} -\typicalityParam)/2$ small enough to satisfy (\ref{proof:resdirect-renyidiv-bound}), then $\proofconstantTwo$ such that $\proofconstantOne < \proofconstantTwo < (\codebookRate - \mutualInformation{\channelIn}{\channelOut} -\typicalityParam)/2$, and finally we choose $\finalconstOne < \proofconstantOne$ and $\finalconstTwo < \min(\codebookRate - \proofconstantOne, \proofconstantTwo-\proofconstantOne)$. With these choices, we get (\ref{thm:resdirect-probstatement}) for all sufficiently large $\codebookBlocklength$, thereby concluding the proof.
\end{shownto}
\begin{shownto}{conference}
Considering (\ref{split-typical-atypical}), (\ref{proof:firstorder-atypical-probability-statement}) and (\ref{typical-bound-final-result}), an application of the union bound yields an upper bound for
$
\codebookPMeasure_\codebook \left(
  \totalvariation{\codebookPMeasure_{\channelOut^\codebookBlocklength | \codebook} - \channelPMeasure_{\channelOut^\codebookBlocklength}}
  >
  3\exp(-\codebookBlocklength\proofconstantOne)
\right)
$.
This bound implies (\ref{thm:resdirect-probstatement}) for all sufficiently large $\codebookBlocklength$ if we choose $\typicalityParam < \codebookRate - \mutualInformation{\channelIn}{\channelOut}$, then $\proofconstantOne < (\codebookRate - \mutualInformation{\channelIn}{\channelOut} -\typicalityParam)/2$ small enough to satisfy the previous condition necessary to ensure (\ref{typical-bound-final-result}), then $\proofconstantTwo$ such that $\proofconstantOne < \proofconstantTwo < (\codebookRate - \mutualInformation{\channelIn}{\channelOut} -\typicalityParam)/2$, then $\finalconstOne < \proofconstantOne$ and finally $\finalconstTwo < \min(\codebookRate - \proofconstantOne, \proofconstantTwo-\proofconstantOne)$.
\end{shownto}
\end{proof}
\begin{shownto}{conference}
\begin{remark}
\label{remark:firstorder}
The requirement that the moment-generating function of the information density exists is needed only to ensure the doubly exponential convergence in (\ref{thm:resdirect-probstatement}). We can instead put the weaker requirement that $\mutualInformation{\channelIn}{\channelOut}$ exists and is finite and replace (\ref{thm:resdirect-probstatement}) with the weaker version
\begin{align*}
\lim\limits_{\codebookBlocklength \rightarrow \infty}
\codebookPMeasure_\codebook \left(
  \totalvariation{
    \codebookPMeasure_{\channelOut^\codebookBlocklength | \codebook} - \channelPMeasure_{\channelOut^\codebookBlocklength}
  }
  >
  \lemmaconst
\right)
=
0
\end{align*}
for all $\lemmaconst > 0$. This is possible by a slight modification of the proof of Theorem~\ref{theorem:firstorder}, in which we use the law of large numbers to show that
$\channelPMeasure_{\channelIn^\codebookBlocklength, \channelOut^\codebookBlocklength}
  (\channelInAlph^\codebookBlocklength \times \channelOutAlph^\codebookBlocklength \setminus \typicalSet{\typicalityParam})$
vanishes with $\codebookBlocklength \rightarrow \infty$ instead of applying the Chernoff bound.
\end{remark}
\end{shownto}
\begin{shownto}{arxiv}
The existence of the moment-generating function is only needed to ensure the doubly exponential convergence in (\ref{thm:resdirect-probstatement}). In fact, modifying the preceding proof slightly, we can also establish the following variation of Theorem~\ref{theorem:firstorder}.
\begin{theorem}
\label{theorem:firstorder-weakened}
Given a channel $\channel = ((\channelInAlph, \sigmaAlgebraIn), (\channelOutAlph, \sigmaAlgebraOut), \channelKernel)$, an input distribution $\channelPMeasure_\channelIn$ such that $\mutualInformation{\channelIn}{\channelOut}$ exists and is finite, there exist $\finalconstOne > 0$ and $\finalconstTwo > 0$ such that for large enough block lengths $\codebookBlocklength$, the randomized codebook distributions of block length $\codebookBlocklength$ and rate $\codebookRate$ satisfy, for every $\lemmaconst > 0$,
\begin{align}
\label{thm:resdirect-probstatement-weakened}
\lim\limits_{\codebookBlocklength \rightarrow \infty}
\codebookPMeasure_\codebook \left(
  \totalvariation{
    \codebookPMeasure_{\channelOut^\codebookBlocklength | \codebook} - \channelPMeasure_{\channelOut^\codebookBlocklength}
  }
  >
  \lemmaconst
\right)
=
0
\end{align}
\end{theorem}
\begin{proof}
The statement (\ref{thm:resdirect-probstatement}) is proven with an application of the union bound, using as ingredients (\ref{split-typical-atypical}), (\ref{proof:firstorder-atypical-probability-statement}) and (\ref{typical-bound-final-result}). (\ref{split-typical-atypical}) and (\ref{typical-bound-final-result}) do not require that the moment-generating function of the information density exists and moreover, (\ref{typical-bound-final-result}) can be weakened to
\[
\lim\limits_{\codebookBlocklength \rightarrow \infty}
\codebookPMeasure_\codebook \left(
  \int\nolimits_{\channelOutAlph^\codebookBlocklength}
    \positive{
      \RNDerivative{\codebookPMeasure_{1,\codebook}}{\channelPMeasure_{\channelOut^\codebookBlocklength}}(\channelOutAlphElement^\codebookBlocklength)
      -
      1
    } 
  \channelPMeasure_{\channelOut^\codebookBlocklength}(d\channelOutAlphElement^\codebookBlocklength)
>
\frac{\lemmaconst}{2}
\right)
=
0
\]
for any $\lemmaconst > 0$. In order to find a suitable replacement for (\ref{proof:firstorder-atypical-probability-statement}), we consider
\begin{multline*}
\Expectation_{\codebookPMeasure_\codebook}(\codebookPMeasure_{2,\codebook}(\channelOutAlph^\codebookBlocklength))
=
\channelPMeasure_{\channelIn^\codebookBlocklength, \channelOut^\codebookBlocklength}
  (\channelInAlph^\codebookBlocklength \times \channelOutAlph^\codebookBlocklength \setminus \typicalSet{\typicalityParam})
=
\\
\channelPMeasure_{\channelIn^\codebookBlocklength, \channelOut^\codebookBlocklength}
  \left(\left\{
    (\channelInAlphElement^\codebookBlocklength, \channelOutAlphElement^\codebookBlocklength):
    \frac{1}{\codebookBlocklength}
    \sum\limits_{\blockIndex=1}^\codebookBlocklength
      \informationDensity{\channelInAlphElement_\blockIndex}{\channelOutAlphElement_\blockIndex}
      -
      \mutualInformation{\channelIn}{\channelOut}
      >
      \typicalityParam
  \right\}\right)
\end{multline*}
and note that by the law of large numbers, it vanishes for any $\typicalityParam > 0$ as $\codebookBlocklength$ tends to infinity. So by Markov's inequality
\begin{align*}
\codebookPMeasure_\codebook\left(\codebookPMeasure_{2,\codebook}(\channelOutAlph^\codebookBlocklength) > \frac{\lemmaconst}{2}\right)
\leq
2
\frac{\Expectation_{\codebookPMeasure_\codebook}(\codebookPMeasure_{2,\codebook}(\channelOutAlph^\codebookBlocklength))}
     {\lemmaconst}
\end{align*}
also vanishes for any $\lemmaconst > 0$. Thus, applying the union bound similarly as in the proof of Theorem~\ref{theorem:firstorder}, we can derive (\ref{thm:resdirect-probstatement-weakened}).
\end{proof}
\end{shownto}
\begin{shownto}{arxiv}
\begin{proof}[Proof of Theorem~\ref{theorem:secondorder}]
We consider the typical set as defined in~(\ref{def:typicalset}), with
\begin{align}
\label{def:secondorder-typicalityparam}
\typicalityParam
:=
\sqrt{\frac{\channelDispersion}{\codebookBlocklength}}
\normalcdfComplementInverse(\secondOrderTVdiff)
+
\secondOrderParamD
\frac{\log \codebookBlocklength}{\codebookBlocklength}.
\end{align}

In preparation for bounding the atypical term in~(\ref{split-typical-atypical}), we observe
\begin{align}
&\hphantom{{}={}}
\channelPMeasure_{\channelIn^\codebookBlocklength, \channelOut^\codebookBlocklength}(
  \channelInAlph^\codebookBlocklength \times \channelOutAlph^\codebookBlocklength \setminus \typicalSet{\typicalityParam}
)
\nonumber
\\
&=
\channelPMeasure_{\channelIn^\codebookBlocklength, \channelOut^\codebookBlocklength}\left(
  \left\{
    (\channelInAlphElement^\codebookBlocklength,\channelOutAlphElement^\codebookBlocklength)
    :
    \informationDensity{\channelInAlphElement^\codebookBlocklength}{\channelOutAlphElement^\codebookBlocklength}/\codebookBlocklength
    >
    \mutualInformation{\channelIn}{\channelOut} + \typicalityParam
  \right\}
\right)
\nonumber
\\
\label{proof:secondorder-atypical-subst}
&\begin{aligned}
=
\channelPMeasure_{\channelIn^\codebookBlocklength, \channelOut^\codebookBlocklength}\Bigg(
  \Bigg\{
    &(\channelInAlphElement^\codebookBlocklength,\channelOutAlphElement^\codebookBlocklength)
    :
    \sum\limits_{\generalSummationIndex=1}^\codebookBlocklength
      \frac{1}{\codebookBlocklength}
      \Big(
        \informationDensity{\channelInAlphElement_\generalSummationIndex}{\channelOutAlphElement_\generalSummationIndex}
        -
        \mutualInformation{\channelIn}{\channelOut}
      \Big)
      \sqrt{\frac{\codebookBlocklength}{\channelDispersion}}
    \\
    &>
    \normalcdfComplementInverse(\secondOrderTVdiff)
    +
    \secondOrderParamD
    \frac{\log \codebookBlocklength}{\sqrt{\codebookBlocklength\channelDispersion}}
  \Bigg\}
\Bigg)
\end{aligned}
\\
\label{proof:secondorder-atypical-berry-esseen}
&\leq
\normalcdfComplement\left(
  \normalcdfComplementInverse(\secondOrderTVdiff)
  +
  \secondOrderParamD
  \frac{\log \codebookBlocklength}{\sqrt{\codebookBlocklength\channelDispersion}}
\right)
+
\frac{\channelThirdMoment}
     {\channelDispersion^{\frac{3}{2}} \sqrt{\codebookBlocklength}}
=
\lemmaexpectation,
\end{align}
where (\ref{proof:secondorder-atypical-subst}) follows by substituting (\ref{def:secondorder-typicalityparam}) and (\ref{proof:secondorder-atypical-berry-esseen}) by the Berry-Esseen Theorem. Knowing this, we apply Lemma~\ref{lemma:atypical} with $\lemmaconst := 1/\sqrt{\codebookBlocklength}$ and get
\begin{align}
\label{proof:secondorder-atypicalbound}
\codebookPMeasure_\codebook\big(
  \codebookPMeasure_{2,\codebook}(\channelOutAlph^\codebookBlocklength)
  >
  \lemmaexpectation(1+1/\sqrt{\codebookBlocklength})
\big)
\leq
\exp\left(
  -\frac{1}{3}
  \codebookBlocklength
  \lemmaexpectation
  \exp(\codebookBlocklength\codebookRate)
\right).
\end{align}
In order to get a bound for the typical part of~(\ref{split-typical-atypical}), we apply Lemma~\ref{lemma:typical} with $\chernoffParam := \exp(\frac{\codebookBlocklength}{2}(\codebookRate-\mutualInformation{\channelIn}{\channelOut}-\typicalityParam))$ and $\lemmaconst := 1/\sqrt{\codebookBlocklength}$, which yields
\begin{align}
&\hphantom{{}={}}
\codebookPMeasure_\codebook\left(
  \Expectation_{\channelPMeasure_{\channelOut^\codebookBlocklength}}
    \positive{
      \RNDerivative{\codebookPMeasure_{1,\codebook}}{\channelPMeasure_{\channelOut^\codebookBlocklength}}(\channelOutAlphElement^\codebookBlocklength)
      -
      1
    } 
  >
  1/\sqrt{\codebookBlocklength}
\right)
\\
\label{proof:secondorder-typicalbound-lemmaapplication}
&\begin{aligned}
\leq
&\left(
  1
  +
  \sqrt{\frac{3\pi}{2}}
  \exp\left(
  \frac{3}{4}
\right) + \exp\left(
    -\frac{\codebookBlocklength}{2}(\codebookRate - \mutualInformation{\channelIn}{\channelOut} - \typicalityParam)
  \right)
\right)
\\
&\cdot\exp\left(
  -\frac{1}{\sqrt{n}}
  \exp\left(
    \frac{1}{2}
    \codebookBlocklength(\codebookRate-\mutualInformation{\channelIn}{\channelOut}-\typicalityParam)
  \right)
\right)
\end{aligned}
\\
\label{proof:secondorder-typicalbound-subst}
&\begin{aligned}
=
&\left(
  1
  +
  \sqrt{\frac{3\pi}{2}}
  \exp\left(
    \frac{3}{4}
  \right)
  +
  \codebookBlocklength^{-\frac{1}{2}(\secondOrderParamC-\secondOrderParamD)}
\right)
\exp\left(
  -\codebookBlocklength^{
    \frac{1}{2}
    (\secondOrderParamC - \secondOrderParamD - 1)
  }
\right)
\end{aligned}
\\
\label{proof:secondorder-typicalbound-final}
&\leq
\left(
  \sqrt{\frac{3\pi}{2}}
  \exp\left(
    \frac{3}{4}
  \right)
  +
  \frac{7}{6}
\right)
\exp\left(
  -\codebookBlocklength^{
    \frac{1}{2}
    (\secondOrderParamC - \secondOrderParamD - 1)
  }
\right),
\end{align}
where (\ref{proof:secondorder-typicalbound-lemmaapplication}) is the application of Lemma~\ref{lemma:typical} taking into account the condition
$\codebookBlocklength^{(\secondOrderParamC-\secondOrderParamD)/2} \geq 6$
and (\ref{proof:secondorder-typicalbound-subst}) follows by substituting (\ref{def:secondorder-rate}) and (\ref{def:secondorder-typicalityparam}). The inequality (\ref{proof:secondorder-typicalbound-final}) holds because of the condition
$\codebookBlocklength^{(\secondOrderParamC-\secondOrderParamD)/2} \geq 6$
.

Finally, we arrive at (\ref{theorem:secondorder-probstatement}) by combining (\ref{proof:secondorder-atypicalbound}) and (\ref{proof:secondorder-typicalbound-final}) using the union bound.
\end{proof}
\end{shownto}

\section{Converse Result}
\label{section:converse}

The main result of this section is Theorem~\ref{theorem:converse}, the converse result for resolvability of continuous channels. We first prove Lemma~\ref{lemma:converse-finite}, a version of the theorem in which only a finite output alphabet is considered, and then show how the statement can be generalized to continuous alphabets by looking at a sequence of discrete approximations of the channel.

\begin{theorem}
\label{theorem:converse}
Let $((\channelInAlph, \sigmaAlgebraIn), (\channelOutAlph, \sigmaAlgebraOut), \channelKernel)$ be a channel such that $\channelInAlph$ is compact and for each $\alphSubset \subseteq \channelOutAlph$, $\channelInAlphElement \mapsto \channelKernel(\channelInAlphElement, \alphSubset)$ is a continuous mapping. Let $\channelPMeasure_\channelOut$ be an output distribution and $(\codebook_\blocklengthIndex)_{\blocklengthIndex \geq 1}$ a sequence of codebooks with strictly increasing block lengths $\codebookBlocklength_\blocklengthIndex$ and fixed rate $\codebookRate$ such that
$
\totalvariation{
  \codebookPMeasure_{\channelOut^{\codebookBlocklength_\blocklengthIndex} | \codebook_{\blocklengthIndex}}
  -
  \channelPMeasure_{\channelOut^{\codebookBlocklength_\blocklengthIndex}}
}
=
\lemmaconst_\blocklengthIndex \leq 1/4
$
with $\lemmaconst_\blocklengthIndex \rightarrow 0$. Then there is a joint probability measure $\channelPMeasure_{\channelIn,\channelOut}$ with marginal $\channelPMeasure_{\channelOut}$ for $\channelOut$ such that $\channelPMeasure_{\channelIn}$ induces $\channelPMeasure_{\channelOut}$ through $\channelKernel$ and
$
\mutualInformationWrt{\channelPMeasure_{\channelIn,\channelOut}}{\channelIn}{\channelOut} \leq \codebookRate
$.
\end{theorem}
\begin{lemma}
\label{lemma:converse-finite}
Let $((\channelInAlph, \sigmaAlgebraIn), (\channelOutAlph, \sigmaAlgebraOut), \channelKernel)$ be a channel such that $\channelInAlph$ is compact, $(\channelOutAlph, \sigmaAlgebraOut)$ is a finite discrete space and $\channelInAlphElement \mapsto \channelKernel(\channelInAlphElement, \cdot)$ is a continuous mapping from $\channelInAlph$ to the probability measures on $\channelOutAlph$. Let $\channelPMeasure_\channelOut$ be an output distribution, and $(\codebook_\blocklengthIndex)_{\blocklengthIndex \geq 1}$ be a sequence of codebooks with strictly increasing block lengths $\codebookBlocklength_\blocklengthIndex$ and fixed rate $\codebookRate$ such that
\begin{align}
\label{lemma:converse-finite-totvar}
\totalvariation{
  \codebookPMeasure_{\channelOut^{\codebookBlocklength_\blocklengthIndex} | \codebook_{\blocklengthIndex}}
  -
  \channelPMeasure_{\channelOut^{\codebookBlocklength_\blocklengthIndex}}
}
=
\lemmaconst_\blocklengthIndex \leq \frac{1}{4}
\end{align}
with $\lemmaconst_\blocklengthIndex \rightarrow 0$. Define
\begin{align}
\channelPMeasure_\channelIn^{(\blocklengthIndex)}
:=
\frac{1}{\codebookBlocklength_\blocklengthIndex}
\sum_{\blockIndex = 1}^{\codebookBlocklength_\blocklengthIndex}
  \codebookPMeasure_{\channelIn_\blockIndex | \codebook_\blocklengthIndex}
\end{align}
and $\channelPMeasure_{\channelIn, \channelOut}^{(\blocklengthIndex)}$ induced by $\channelPMeasure_\channelIn^{(\blocklengthIndex)}$ through $\channelKernel$. Then there is a strictly increasing sequence $(\blocklengthIndex_\blocklengthSubsequenceIndex)_{\blocklengthSubsequenceIndex \geq 1}$ such that $\channelPMeasure_{\channelIn, \channelOut}^{(\blocklengthIndex_\blocklengthSubsequenceIndex)}$ converges weakly to some $\channelPMeasure_{\channelIn, \channelOut}$, the marginal $\channelPMeasure_\channelIn$ induces the marginal $\channelPMeasure_\channelOut$ through $\channelKernel$ and $\mutualInformationWrt{\channelPMeasure_{\channelIn,\channelOut}}{\channelIn}{\channelOut} \leq \codebookRate$.
\end{lemma}
\begin{remark}
The compactness of $\channelInAlph$ and the continuity of $\channelInAlphElement \mapsto \channelKernel(\channelInAlphElement,\alphSubset)$ are technical conditions to ensure the convergence in Lemma~\ref{lemma:converse-finite}. It should be a subject of further research to explore whether these conditions can be dropped by using a more involved discretization or a wholly different technique.
\end{remark}
\begin{proof}[Proof of Lemma~\ref{lemma:converse-finite}]
Since $\channelInAlph$ is compact, the space of measures $\generalPMeasure$ on $(\channelInAlph, \sigmaAlgebraIn)$ such that $\generalPMeasure(\channelInAlph) \leq 1$ endowed with the weak topology is compact~\cite[Corollary 31.3]{BauerMeasure}. Therefore, $(\channelPMeasure_\channelIn^{(\blocklengthIndex)})_{\blocklengthIndex \geq 1}$ must have a convergent subsequence, or, put differently, there is a strictly increasing sequence $(\blocklengthIndex_\blocklengthSubsequenceIndex)_{\blocklengthSubsequenceIndex \geq 1}$ such that $(\channelPMeasure_\channelIn^{(\blocklengthIndex_\blocklengthSubsequenceIndex)})_{\blocklengthSubsequenceIndex \geq 1}$ converges weakly to some $\channelPMeasure_{\hat{\channelIn}}$. By~\cite[Theorem A.5.9]{dupuis2011weak}, $\channelPMeasure_{\channelIn, \channelOut}^{(\blocklengthIndex_\blocklengthSubsequenceIndex)}$ converges weakly to some $\channelPMeasure_{\hat{\channelIn}, \hat{\channelOut}}$ and the marginal $\channelPMeasure_{\hat{\channelIn}}$ induces the marginal $\channelPMeasure_{\hat{\channelOut}}$ through $\channelKernel$.
\begin{shownto}{arxiv}
We note
\begin{align*}
&\hphantom{{}={}}
\totalvariation{
  \codebookPMeasure_{\channelOut_\blockIndex | \codebook_\blocklengthIndex}
  -
  \channelPMeasure_\channelOut
}
\\
&=
\frac{1}{2}
\sum\limits_{\channelOutAlphElement \in \channelOutAlph} \absolute{
  \codebookPMeasure_{\channelOut_\blockIndex | \codebook_\blocklengthIndex}(\{\channelOutAlphElement\})
  -
  \channelPMeasure_\channelOut(\{\channelOutAlphElement\})
}
\\
&=
\frac{1}{2}
\sum\limits_{\channelOutAlphElement \in \channelOutAlph} \absolute{
  \sum_{\substack{\channelOutAlphElement^{\codebookBlocklength_\blocklengthIndex} \in \channelOutAlph^{\codebookBlocklength_\blocklengthIndex} \\
                 \channelOutAlphElement_\blockIndex = \channelOutAlphElement}}
    \left(
      \codebookPMeasure_{\channelOut^{\codebookBlocklength_\blocklengthIndex} | \codebook_\blocklengthIndex}(\{\channelOutAlphElement^{\codebookBlocklength_\blocklengthIndex}\})
      -
      \channelPMeasure_{\channelOut^{\codebookBlocklength_\blocklengthIndex}}(\{\channelOutAlphElement^{\codebookBlocklength_\blocklengthIndex}\})
    \right)
}
\\
&\leq
\frac{1}{2}
\sum\limits_{\channelOutAlphElement^{\codebookBlocklength_\blocklengthIndex} \in \channelOutAlph^{\codebookBlocklength_\blocklengthIndex}} \absolute{
  \codebookPMeasure_{\channelOut^{\codebookBlocklength_\blocklengthIndex} | \codebook_\blocklengthIndex}(\{\channelOutAlphElement^{\codebookBlocklength_\blocklengthIndex}\})
  -
  \channelPMeasure_{\channelOut^{\codebookBlocklength_\blocklengthIndex}}(\{\channelOutAlphElement^{\codebookBlocklength_\blocklengthIndex}\})
}
\\
&=
\totalvariation{
  \codebookPMeasure_{\channelOut^{\codebookBlocklength_\blocklengthIndex} | \codebook_{\blocklengthIndex}}
  -
  \channelPMeasure_{\channelOut^{\codebookBlocklength_\blocklengthIndex}}
}
\end{align*}
\end{shownto}
\begin{shownto}{conference}
We note that, by the triangle inequality,
$
\totalvariation{
  \codebookPMeasure_{\channelOut_\blockIndex | \codebook_\blocklengthIndex}
  -
  \channelPMeasure_\channelOut
}
\leq
\totalvariation{
  \codebookPMeasure_{\channelOut^{\codebookBlocklength_\blocklengthIndex} | \codebook_{\blocklengthIndex}}
  -
  \channelPMeasure_{\channelOut^{\codebookBlocklength_\blocklengthIndex}}
}
$
for any $\blockIndex$, and therefore
\end{shownto}
\begin{shownto}{arxiv}
and therefore
\begin{align}
\totalvariation{
  \channelPMeasure_\channelOut^{(\blocklengthIndex)}
  -
  \channelPMeasure_\channelOut
}
&=
\totalvariation{
  \frac{1}{\codebookBlocklength_\blocklengthIndex}
  \sum_{\blockIndex = 1}^{\codebookBlocklength_\blocklengthIndex}
    \codebookPMeasure_{\channelOut_\blockIndex | \codebook_\blocklengthIndex}
  -
  \channelPMeasure_\channelOut
}
\nonumber
\\
&\leq
\frac{1}{\codebookBlocklength_\blocklengthIndex}
\sum_{\blockIndex = 1}^{\codebookBlocklength_\blocklengthIndex}
\totalvariation{
  \codebookPMeasure_{\channelOut_\blockIndex | \codebook_\blocklengthIndex}
  -
  \channelPMeasure_\channelOut
}
\nonumber
\\
&\leq
\totalvariation{
  \codebookPMeasure_{\channelOut^{\codebookBlocklength_\blocklengthIndex} | \codebook_{\blocklengthIndex}}
  -
  \channelPMeasure_{\channelOut^{\codebookBlocklength_\blocklengthIndex}}
}.
\label{proof:converse-finite-totvar-bound}
\end{align}
\end{shownto}
\begin{shownto}{conference}
\begin{align}
&\hphantom{{}={}}
\totalvariation{
  \channelPMeasure_\channelOut^{(\blocklengthIndex)}
  -
  \channelPMeasure_\channelOut
}
=
\totalvariation{
  \frac{1}{\codebookBlocklength_\blocklengthIndex}
  \sum_{\blockIndex = 1}^{\codebookBlocklength_\blocklengthIndex}
    \codebookPMeasure_{\channelOut_\blockIndex | \codebook_\blocklengthIndex}
  -
  \channelPMeasure_\channelOut
}
\nonumber
\\
&\leq
\frac{1}{\codebookBlocklength_\blocklengthIndex}
\sum_{\blockIndex = 1}^{\codebookBlocklength_\blocklengthIndex}
\totalvariation{
  \codebookPMeasure_{\channelOut_\blockIndex | \codebook_\blocklengthIndex}
  -
  \channelPMeasure_\channelOut
}
\leq
\totalvariation{
  \codebookPMeasure_{\channelOut^{\codebookBlocklength_\blocklengthIndex} | \codebook_{\blocklengthIndex}}
  -
  \channelPMeasure_{\channelOut^{\codebookBlocklength_\blocklengthIndex}}
}.
\label{proof:converse-finite-totvar-bound}
\end{align}
\end{shownto}
So $\channelPMeasure_\channelOut^{(\blocklengthIndex)}$ converges to $\channelPMeasure_\channelOut$ in total variation and thus, in particular, weakly. Moreover, we have $\channelPMeasure_\channelOut = \channelPMeasure_{\hat{\channelOut}}$ because marginalization is a continuous operation under the weak topology, so we can write $\channelPMeasure_{\channelIn,\channelOut}$ instead of $\channelPMeasure_{\hat{\channelIn},\hat{\channelOut}}$. We further observe
\begin{shownto}{arxiv}
\begin{align}
&\hphantom{{}={}}
\codebookBlocklength_\blocklengthIndex \codebookRate
\geq
\entropyWrt{\codebookPMeasure_{\channelIn^{\codebookBlocklength_\blocklengthIndex} | \codebook_{\blocklengthIndex}}}{\channelIn}
\geq
\mutualInformationWrt{\codebookPMeasure_{\channelIn^{\codebookBlocklength_\blocklengthIndex}, \channelOut^{\codebookBlocklength_\blocklengthIndex} | \codebook_{\blocklengthIndex}}}{\channelIn^{\codebookBlocklength_\blocklengthIndex}}{\channelOut^{\codebookBlocklength_\blocklengthIndex}}
\nonumber
\\
&=
\sum_{\substack{\channelInAlphElement^{\codebookBlocklength_\blocklengthIndex} \in \channelInAlph^{\codebookBlocklength_\blocklengthIndex} \\
                \channelOutAlphElement^{\codebookBlocklength_\blocklengthIndex} \in \channelOutAlph^{\codebookBlocklength_\blocklengthIndex}}}
  \codebookPMeasure_{\channelIn^{\codebookBlocklength_\blocklengthIndex}, \channelOut^{\codebookBlocklength_\blocklengthIndex} | \codebook_\blocklengthIndex}
    (\{(\channelInAlphElement^{\codebookBlocklength_\blocklengthIndex}, \channelOutAlphElement^{\codebookBlocklength_\blocklengthIndex})\})
  \log
    \frac{\kernelPower{\channelKernel}{\codebookBlocklength_\blocklengthIndex}(\channelInAlphElement^{\codebookBlocklength_\blocklengthIndex}, \{\channelOutAlphElement^{\codebookBlocklength_\blocklengthIndex}\})}
         {\codebookPMeasure_{\channelOut^{\codebookBlocklength_\blocklengthIndex} | \codebook_\blocklengthIndex}(\{\channelOutAlphElement^{\codebookBlocklength_\blocklengthIndex}\})}
\nonumber
\\
&\begin{aligned}
=\>
&\sum_{\blockIndex=1}^{\codebookBlocklength_\blocklengthIndex}
\sum_{\substack{\channelInAlphElement \in \channelInAlph \\
                \channelOutAlphElement \in \channelOutAlph}}
  \codebookPMeasure_{\channelIn_\blockIndex, \channelOut_\blockIndex | \codebook_\blocklengthIndex}
    (\{(\channelInAlphElement, \channelOutAlphElement)\})
  \log \channelKernel(\channelInAlphElement, \{\channelOutAlphElement\})
\\
&+
\entropyWrt{\codebookPMeasure_{\channelOut^{\codebookBlocklength_\blocklengthIndex} | \codebook_\blocklengthIndex}}{\channelOut^{\codebookBlocklength_\blocklengthIndex}}
\end{aligned}
\nonumber
\\
&\begin{aligned}
=\>
&\codebookBlocklength_\blocklengthIndex
\sum_{\substack{\channelInAlphElement \in \channelInAlph \\
                \channelOutAlphElement \in \channelOutAlph}}
  \channelPMeasure_{\channelIn, \channelOut}^{(\blocklengthIndex)}(\{\channelInAlphElement, \channelOutAlphElement\})
  \log
    \frac{\channelKernel(\channelInAlphElement, \{\channelOutAlphElement\})}
         {\channelPMeasure_{\channelOut}^{(\blocklengthIndex)}(\{\channelOutAlphElement\})}
\\
&+
\codebookBlocklength_\blocklengthIndex
\sum_{\substack{\channelInAlphElement \in \channelInAlph \\
                \channelOutAlphElement \in \channelOutAlph}}
  \channelPMeasure_{\channelIn, \channelOut}^{(\blocklengthIndex)}(\{\channelInAlphElement, \channelOutAlphElement\})
  \log
    \channelPMeasure_{\channelOut}^{(\blocklengthIndex)}(\{\channelOutAlphElement\})
+
\entropyWrt{\codebookPMeasure_{\channelOut^{\codebookBlocklength_\blocklengthIndex} | \codebook_\blocklengthIndex}}{\channelOut^{\codebookBlocklength_\blocklengthIndex}}
\end{aligned}
\nonumber
\\
&=
\codebookBlocklength_\blocklengthIndex
\mutualInformationWrt{\channelPMeasure_{\channelIn, \channelOut}^{(\blocklengthIndex)}}{\channelIn}{\channelOut}
-
\codebookBlocklength_\blocklengthIndex
\entropyWrt{\channelPMeasure_{\channelOut}^{(\blocklengthIndex)}}{\channelOut}
+
\entropyWrt{\codebookPMeasure_{\channelOut^{\codebookBlocklength_\blocklengthIndex} | \codebook_\blocklengthIndex}}{\channelOut^{\codebookBlocklength_\blocklengthIndex}}
\nonumber
\\
\label{proof:converse-finite-totvar-application}
&\geq
\codebookBlocklength_\blocklengthIndex
\mutualInformationWrt{\channelPMeasure_{\channelIn, \channelOut}^{(\blocklengthIndex)}}{\channelIn}{\channelOut}
+
\frac{1}{2}
\lemmaconst_\blocklengthIndex
\log\frac{\lemmaconst_\blocklengthIndex}{2 \cardinality{\channelOutAlph^{\codebookBlocklength_\blocklengthIndex}}}
+
\frac{1}{2}
{\codebookBlocklength_\blocklengthIndex}
\lemmaconst_\blocklengthIndex
\log\frac{\lemmaconst_\blocklengthIndex}{2 \cardinality{\channelOutAlph}}
\\
&\geq
\codebookBlocklength_\blocklengthIndex
\left(
  \mutualInformationWrt{\channelPMeasure_{\channelIn, \channelOut}^{(\blocklengthIndex)}}{\channelIn}{\channelOut}
  +
  \lemmaconst_\blocklengthIndex
  \log\frac{\lemmaconst_\blocklengthIndex}{2 \cardinality{\channelOutAlph}}
\right)
\nonumber,
\end{align}
\end{shownto}
\begin{shownto}{conference}
\begin{align}
&\hphantom{{}={}}
\codebookBlocklength_\blocklengthIndex \codebookRate
\geq
\entropyWrt{\codebookPMeasure_{\channelIn^{\codebookBlocklength_\blocklengthIndex} | \codebook_{\blocklengthIndex}}}{\channelIn}
\geq
\mutualInformationWrt{\codebookPMeasure_{\channelIn^{\codebookBlocklength_\blocklengthIndex}, \channelOut^{\codebookBlocklength_\blocklengthIndex} | \codebook_{\blocklengthIndex}}}{\channelIn^{\codebookBlocklength_\blocklengthIndex}}{\channelOut^{\codebookBlocklength_\blocklengthIndex}}
\nonumber
\\
&\begin{aligned}
=\>
&\sum_{\blockIndex=1}^{\codebookBlocklength_\blocklengthIndex}
\sum_{\substack{\channelInAlphElement \in \channelInAlph \\
                \channelOutAlphElement \in \channelOutAlph}}
  \codebookPMeasure_{\channelIn_\blockIndex, \channelOut_\blockIndex | \codebook_\blocklengthIndex}
    (\{(\channelInAlphElement, \channelOutAlphElement)\})
  \log \channelKernel(\channelInAlphElement, \{\channelOutAlphElement\})
\\
&+
\entropyWrt{\codebookPMeasure_{\channelOut^{\codebookBlocklength_\blocklengthIndex} | \codebook_\blocklengthIndex}}{\channelOut^{\codebookBlocklength_\blocklengthIndex}}
\end{aligned}
\nonumber
\\
&\begin{aligned}
=\>
&\codebookBlocklength_\blocklengthIndex
\sum_{\substack{\channelInAlphElement \in \channelInAlph \\
                \channelOutAlphElement \in \channelOutAlph}}
  \channelPMeasure_{\channelIn, \channelOut}^{(\blocklengthIndex)}(\{\channelInAlphElement, \channelOutAlphElement\})
  \log
    \frac{\channelKernel(\channelInAlphElement, \{\channelOutAlphElement\})}
         {\channelPMeasure_{\channelOut}^{(\blocklengthIndex)}(\{\channelOutAlphElement\})}
\\
&+
\codebookBlocklength_\blocklengthIndex
\sum_{\substack{\channelInAlphElement \in \channelInAlph \\
                \channelOutAlphElement \in \channelOutAlph}}
  \channelPMeasure_{\channelIn, \channelOut}^{(\blocklengthIndex)}(\{\channelInAlphElement, \channelOutAlphElement\})
  \log
    \channelPMeasure_{\channelOut}^{(\blocklengthIndex)}(\{\channelOutAlphElement\})
+
\entropyWrt{\codebookPMeasure_{\channelOut^{\codebookBlocklength_\blocklengthIndex} | \codebook_\blocklengthIndex}}{\channelOut^{\codebookBlocklength_\blocklengthIndex}}
\end{aligned}
\nonumber
\\
&=
\codebookBlocklength_\blocklengthIndex
\mutualInformationWrt{\channelPMeasure_{\channelIn, \channelOut}^{(\blocklengthIndex)}}{\channelIn}{\channelOut}
-
\codebookBlocklength_\blocklengthIndex
\entropyWrt{\channelPMeasure_{\channelOut}^{(\blocklengthIndex)}}{\channelOut}
+
\entropyWrt{\codebookPMeasure_{\channelOut^{\codebookBlocklength_\blocklengthIndex} | \codebook_\blocklengthIndex}}{\channelOut^{\codebookBlocklength_\blocklengthIndex}}
\nonumber
\\
\label{proof:converse-finite-totvar-application}
&\geq
\codebookBlocklength_\blocklengthIndex
\left(
  \mutualInformationWrt{\channelPMeasure_{\channelIn, \channelOut}^{(\blocklengthIndex)}}{\channelIn}{\channelOut}
  +
  \lemmaconst_\blocklengthIndex
  \log\frac{\lemmaconst_\blocklengthIndex}{2 \cardinality{\channelOutAlph}}
\right),
\end{align}
\end{shownto}
where (\ref{proof:converse-finite-totvar-application}) is an application of~\cite[Lemma 2.7]{CsiszarInformation}\footnote{Note that the definition of the variational distance in~\cite{CsiszarInformation} differs from the one used in this paper by a factor of $1/2$.} taking into consideration (\ref{lemma:converse-finite-totvar}) and (\ref{proof:converse-finite-totvar-bound}). Thus, in particular, because the mutual information is lower semicontinuous in the weak topology \cite[Theorem 1]{PosnerRandom}, we can conclude
\[
\mutualInformationWrt{\channelPMeasure_{\channelIn, \channelOut}}{\channelIn}{\channelOut}
\leq
\liminf\limits_{\blocklengthIndex \rightarrow \infty}
  \mutualInformationWrt{\channelPMeasure_{\channelIn, \channelOut}^{(\blocklengthIndex)}}{\channelIn}{\channelOut}
\leq
\codebookRate.
\qedhere
\]
\end{proof}

\begin{proof}[Proof of Theorem~\ref{theorem:converse}]
Pick an increasing sequence $(\sigmaAlgebraOut_\channelDiscretizationIndex)_{\channelDiscretizationIndex \geq 1}$ of finite algebras on $\channelOutAlph$ such that $\bigcup_{\channelDiscretizationIndex \geq 1} \sigmaAlgebraOut_\channelDiscretizationIndex$ generates $\sigmaAlgebraOut$.
Recursively construct sequences $(\blocklengthIndex_\blocklengthSubsequenceIndex^{(\channelDiscretizationIndex)})_{\blocklengthSubsequenceIndex \geq 1}$ for each $\channelDiscretizationIndex \geq 0$. Set $\blocklengthIndex_\blocklengthSubsequenceIndex^{(0)} := \blocklengthSubsequenceIndex$. In order to construct $(\blocklengthIndex_\blocklengthSubsequenceIndex^{(\channelDiscretizationIndex)})_{\blocklengthSubsequenceIndex \geq 1}$ for $\channelDiscretizationIndex > 0$, first define a discrete finite alphabet
$
\channelOutAlph_\channelDiscretizationIndex := \{
  \channelOutAlphElement \subseteq \channelOutAlph
  :
  \channelOutAlphElement
  \text{ is an atom of }
  \sigmaAlgebraOut_\channelDiscretizationIndex
\}
$
and note that any probability measure $\generalPMeasure$ on $\channelOutAlph$ induces a probability measure $\generalPMeasure^{(\channelDiscretizationIndex)}$ on $\channelOutAlph_\channelDiscretizationIndex$ via
$
\generalPMeasure^{(\channelDiscretizationIndex)}(\alphSubset) := \generalPMeasure\left(\bigcup \alphSubset\right)
$
and conversely, $\generalPMeasure^{(\channelDiscretizationIndex)}$ can be seen as a probability measure on $(\channelOutAlph, \sigmaAlgebraOut_\channelDiscretizationIndex)$ by assigning to any set in $\sigmaAlgebraOut_\channelDiscretizationIndex$ the sum of the probabilities of all contained atoms, or, put equivalently, $\generalPMeasure^{(\channelDiscretizationIndex)}$ can be seen as the restriction of $\generalPMeasure$ to $\sigmaAlgebraOut_\channelDiscretizationIndex$.

So for each $\channelInAlphElement \in \channelInAlph$, $\channelKernel(\channelInAlphElement, \cdot)$ induces a probability measure $\channelKernel^{(\channelDiscretizationIndex)}(\channelInAlphElement, \cdot)$ on $\channelOutAlph_\channelDiscretizationIndex$ and thus we get a stochastic kernel $\channelKernel^{(\channelDiscretizationIndex)}$ and thereby a channel $((\channelInAlph, \sigmaAlgebraIn), (\channelOutAlph_\channelDiscretizationIndex, \powerset(\channelOutAlph_\channelDiscretizationIndex))  , \channelKernel^{(\channelDiscretizationIndex)})$. $\channelInAlphElement \mapsto \channelKernel^{(\channelDiscretizationIndex)}(\channelInAlphElement, \cdot)$ is a continuous map to the space of probability vectors on $\channelOutAlph_\channelDiscretizationIndex$, because $\channelInAlphElement \mapsto \channelKernel(\channelInAlphElement, \channelOutAlphElement)$ is continuous for each $\channelOutAlphElement \in \channelOutAlph_\channelDiscretizationIndex$. Furthermore, $\channelPMeasure_\channelOut$ induces $\channelPMeasure_\channelOut^{(\channelDiscretizationIndex)}$ in the same way. A codebook $\codebook_\blocklengthIndex$ is also a codebook for the new channel, and we note that also $\codebookPMeasure_{\channelOut^{\codebookBlocklength_\blocklengthIndex} | \codebook_\blocklengthIndex}^{(\channelDiscretizationIndex)}$ is induced by $\codebookPMeasure_{\channelOut^{\codebookBlocklength_\blocklengthIndex} | \codebook_\blocklengthIndex}$ in the same way. So for any $\alphSubset \subseteq \channelOutAlph_\channelDiscretizationIndex$, we have $\codebookPMeasure_{\channelOut^{\codebookBlocklength_\blocklengthIndex} | \codebook_\blocklengthIndex}^{(\channelDiscretizationIndex)}(\alphSubset) = \codebookPMeasure_{\channelOut^{\codebookBlocklength_\blocklengthIndex} | \codebook_\blocklengthIndex}( \bigcup \alphSubset )$ and $\channelPMeasure_\channelOut^{(\channelDiscretizationIndex)}(\alphSubset) = \channelPMeasure_\channelOut(\bigcup\alphSubset)$. Thus,
\begin{shownto}{arxiv}
\begin{align*}
\totalvariation{
  \codebookPMeasure_{\channelOut^{\codebookBlocklength_\blocklengthIndex} | \codebook_\blocklengthIndex}^{(\channelDiscretizationIndex)}
  -
  \channelPMeasure_\channelOut^{(\channelDiscretizationIndex)}
}
&=
\sup\limits_{\alphSubset \subseteq \channelOutAlph_\channelDiscretizationIndex} \left(
  \codebookPMeasure_{\channelOut^{\codebookBlocklength_\blocklengthIndex} | \codebook_\blocklengthIndex}^{(\channelDiscretizationIndex)}(\alphSubset)
  -
  \channelPMeasure_\channelOut^{(\channelDiscretizationIndex)}(\alphSubset)
\right)
\\
&\leq
\sup\limits_{\alphSubset \in \sigmaAlgebraOut} \left(
  \codebookPMeasure_{\channelOut^{\codebookBlocklength_\blocklengthIndex} | \codebook_\blocklengthIndex}(\alphSubset)
  -
  \channelPMeasure_\channelOut(\alphSubset)
\right)
\\
&=
\totalvariation{
  \codebookPMeasure_{\channelOut^{\codebookBlocklength_\blocklengthIndex} | \codebook_\blocklengthIndex}
  -
  \channelPMeasure_\channelOut
}.
\end{align*}
\end{shownto}
\begin{shownto}{conference}
\begin{align*}
&\hphantom{{}={}}
\totalvariation{
  \codebookPMeasure_{\channelOut^{\codebookBlocklength_\blocklengthIndex} | \codebook_\blocklengthIndex}^{(\channelDiscretizationIndex)}
  -
  \channelPMeasure_\channelOut^{(\channelDiscretizationIndex)}
}
=
\sup\limits_{\alphSubset \subseteq \channelOutAlph_\channelDiscretizationIndex} \left(
  \codebookPMeasure_{\channelOut^{\codebookBlocklength_\blocklengthIndex} | \codebook_\blocklengthIndex}^{(\channelDiscretizationIndex)}(\alphSubset)
  -
  \channelPMeasure_\channelOut^{(\channelDiscretizationIndex)}(\alphSubset)
\right)
\\
&\leq
\sup\limits_{\alphSubset \in \sigmaAlgebraOut} \left(
  \codebookPMeasure_{\channelOut^{\codebookBlocklength_\blocklengthIndex} | \codebook_\blocklengthIndex}(\alphSubset)
  -
  \channelPMeasure_\channelOut(\alphSubset)
\right)
=
\totalvariation{
  \codebookPMeasure_{\channelOut^{\codebookBlocklength_\blocklengthIndex} | \codebook_\blocklengthIndex}
  -
  \channelPMeasure_\channelOut
}.
\end{align*}
\end{shownto}

We can therefore apply Lemma~\ref{lemma:converse-finite} to the codebook sequence $(\codebook_{\blocklengthIndex_\blocklengthSubsequenceIndex^{(\channelDiscretizationIndex - 1)}})_{\blocklengthSubsequenceIndex \geq 1}$ to obtain a subsequence $(\blocklengthIndex_\blocklengthSubsequenceIndex^{(\channelDiscretizationIndex)})_{\blocklengthSubsequenceIndex \geq 1}$ of $(\blocklengthIndex_\blocklengthSubsequenceIndex^{(\channelDiscretizationIndex-1)})_{\blocklengthSubsequenceIndex \geq 1}$ such that $\channelPMeasure_{\channelIn, \channelOut}^{(\blocklengthIndex_\blocklengthSubsequenceIndex^{(\channelDiscretizationIndex)})}$ converges weakly to some $\channelPMeasure_{\channelIn, \channelOut}^{(\channelDiscretizationIndex)}$ such that $\channelPMeasure_{\channelIn}^{(\channelDiscretizationIndex)}$ induces $\channelPMeasure_{\channelOut}^{(\channelDiscretizationIndex)}$ through $\channelKernel^{(\channelDiscretizationIndex)}$. Because of the subsequence construction, all the $\channelPMeasure_{\channelIn, \channelOut}^{(\channelDiscretizationIndex)}$ are compatible with each other and seeing them as measures on $(\channelInAlph \times \channelOutAlph, \sigmaAlgebraIn \sigmaAlgebraProduct \sigmaAlgebraOut_\channelDiscretizationIndex)$, we can define a probability measure $\channelPMeasure_{\channelIn, \channelOut}$ on $(\channelInAlph \times \channelOutAlph, \sigmaAlgebraIn \sigmaAlgebraProduct \sigmaAlgebraOut)$ as the unique extension~\cite[Theorem 3.1]{BillingsleyProbability} of
$
\bigcup_{\channelDiscretizationIndex \geq 1}
  \channelPMeasure_{\channelIn, \channelOut}^{(\channelDiscretizationIndex)}
$
to the $\sigma$-algebra $\sigmaAlgebraIn \sigmaAlgebraProduct \sigmaAlgebraOut$. Since also the marginals are compatible (the marginals for $\channelIn$ are even identical), we can apply \cite[Corollary 5.2.3]{gray2011entropy} and obtain $\mutualInformationWrt{\channelPMeasure_{\channelIn, \channelOut}^{(\channelDiscretizationIndex)}}{\channelIn}{\channelOut} \rightarrow \mutualInformationWrt{\channelPMeasure_{\channelIn, \channelOut}}{\channelIn}{\channelOut}$. Since we know from the statement of Lemma~\ref{lemma:converse-finite} that for all $\channelDiscretizationIndex$, $\mutualInformationWrt{\channelPMeasure_{\channelIn, \channelOut}^{(\channelDiscretizationIndex)}}{\channelIn}{\channelOut} \leq \codebookRate$, it follows that $\mutualInformationWrt{\channelPMeasure_{\channelIn, \channelOut}}{\channelIn}{\channelOut} \leq \codebookRate$, completing the proof of the theorem.
\end{proof}

\bibliographystyle{plain}
\bibliography{references}
\end{document}